\documentclass[a4paper,english]{lipics-report}

\usepackage{hyperref}
\usepackage{microtype}
\usepackage[normalem]{ulem}
\usepackage{amssymb}
\usepackage{amsmath,amsfonts}
\usepackage{graphicx}
\usepackage{color}
\usepackage{graphicx}

\newcommand{\sfmtl}{\mathsf{SfrMTL}}
\newcommand{\sfmitl}{\mathsf{SfrMITL}}

\newcommand{\uregm}{\mathsf{URat}}

\newcommand{\tforall}{\overline{\forall}}
\newcommand{\texists}{{\overline{\exists}}}

\newcommand{\regmtl}{\mathsf{RatMTL}}
\newcommand{\fregmtl}{\mathsf{FRatMTL}}
\newcommand{\fregmitl}{\mathsf{FRatMITL}}

\newcommand{\lfr}{\mathsf{lfr}}
\newcommand{\wf}{\mathsf{C}\tiny{\oplus}\mathsf{D}}

\newcommand{\chec}{\mathsf{check}}

\newcommand{\re}{\mathsf{re}}
\newcommand{\reg}{\mathsf{Rat}}
\newcommand{\ureg}{\mathsf{URat}}
\newcommand{\freg}{\mathsf{FRat}}
\newcommand{\greg}{\mathsf{GRat}}

\newcommand{\qtwomlo}{\mathsf{Q2FO}}
\newcommand{\qkmlo}{\mathsf{QkFO}}
\newcommand{\qtwomso}{\mathsf{Q2MSO}}
\newcommand{\qkmso}{\mathsf{QkMSO}}
\newcommand{\qfourmso}{\mathsf{Q4MSO}}
\newcommand{\qfourmlo}{\mathsf{Q4FO}}

\newcommand{\md}{\mathsf{md}}
\newcommand{\Qq}{\mathcal{Q}}

\newcommand{\Tt}{\mathcal{T}}
\newcommand{\Aa}{\mathcal{A}}
\newcommand{\norm}{\mathsf{Norm}}
\newcommand{\Bb}{\mathcal{B}}
\newcommand{\Cc}{\mathcal{C}}
\newcommand{\Dd}{\mathcal{D}}

\newcommand{\Zz}{\mathcal{Z}}
\newcommand{\Ii}{\mathcal{I}}

\newcommand{\regm}{\mathsf{Rat}}
\newcommand{\singl}{\mathsf{single}}
\newcommand{\mitl}{\mathsf{MITL}}

\newcommand{\Ss}{\mathsf{S}}

\newcommand{\F}{\mathsf{F}}

\newcommand{\Nat}{\mathbb N}
\newcommand{\Real}{\mathbb R}

\newcommand{\wB}{\Box^{\mathsf{w}}}

\newcommand{\until}{\:\mathsf{U}}

\newcommand{\R}{\:\mathbb{R}}

\mathchardef\mhyphen="2D
\mathchardef\mhyph="2D

\newcommand{\mtl}{\mathsf{MTL}}

\newcommand{\oomit}[1]{}

\newcommand{\po}{\mathsf{PO}}

\begin{document}
\title{B\"uchi-Kamp Theorems for 1-clock ATA}
\author[1]{S. Krishna}
\author[1]{Khushraj Madnani}
\author[2]{P. K. Pandya}
\affil[1]{IIT Bombay, Mumbai, India\\
\texttt{krishnas,khushraj@cse.iitb.ac.in}}
\affil[2]{TIFR Mumbai, India \\ 
\texttt{pandya@tcs.tifr.res.in}}
\authorrunning{Krishna, Khushraj and Paritosh Pandya}
\maketitle

\abstract
This paper investigates Kamp-like and B\"uchi-like theorems for 1-clock Alternating Timed Automata (1-ATA)
and its natural subclasses. 
A notion of 1-ATA with loop-free-resets is defined. This automaton
class is shown to be expressively equivalent to the temporal logic $\regmtl$ which is $\mathsf{MTL[F_I]}$ extended with a regular expression guarded modality. Moreover, a subclass of  future timed MSO with k-variable-connectivity property is introduced as logic $\qkmso$.
In a Kamp-like result, it is shown that $\regmtl$ is expressively equivalent to $\qkmso$. As our second result, we define a notion of 
conjunctive-disjunctive 1-clock ATA ($\wf$ 1-ATA). We show that $\wf$ 1-ATA with loop-free-resets are expressively equivalent to the sublogic
$\F\regmtl$ of $\regmtl$. Moreover $\F\regmtl$ is expressively equivalent to $\qtwomso$, the two-variable connected fragment of $\qkmso$. The full class of 1-ATA is shown to be expressively equivalent to $\regmtl$ extended with fixed point operators.

  \section{Introduction}
  
  The celebrated Kamp theorem proves expressive equivalence between classical logic and temporal logic over words. 
Equally celebrated are B\"uchi theorems, which  prove equivalence between classical/temporal logic and finite state automata. 
They constitute important results in the theory of logics, automata and their correspondences. 
Unfortunately, such correspondences have been hard to work out for timed automata and timed logics. 
This paper investigates B\"uchi-Kamp like theorems for an important class of timed languages, 
those accepted by 1-clock alternating timed automata (referred to as 1-clock ATA or 1-ATA from here on).

There are several different interpretations of timed logics in the literature. Notable variants are pointwise logics and
continuous timed logics over  finite and infinite timed words \cite{ow08}. 
 Of these, pointwise logic $\mathsf{MTL[U_I]}$ over finite words
has a special place for having decidable satisfiability. In this paper, we focus on only pointwise logics over finite timed words.

1-ATA over finite words are perhaps the largest boolean closed class of timed languages for which emptiness is 
known to be decidable.  Utilizing this fact, Ouaknine and Worrell showed in their seminal work that satisfiability of pointwise $\mathsf{MTL[U_I]}$ over
finite words is decidable, by constructing a language equivalent 1-clock ATA \cite{Ouaknine05}, \cite{igorL}
for a formula of $\mathsf{MTL[U_I]}$. Unfortunately, the logic 
turns out to have much less expressive power than 1-ATA. Indeed $\mathsf{MTL[U_I]}$ can be reduced to partially ordered 1-ATA and is even weaker than the latter. In previous work \cite{mfcs17}, we presented a logic $\sfmtl$\/ which is expressively equivalent for $\po$ 1-clock ATA. 

In a series of papers \cite{fossacs16}, \cite{time14}
 we have investigated  decidable extensions of $\mathsf{MTL[U_I]}$ with increasing expressive power culminating in $\regmtl$ \cite{mfcs17},  but they all turn out to be less expressive than 1-ATA.  Strong B\"uchi-Kamp results have remained elusive. 
In this paper, we now attempt to present some B\"uchi-Kamp like theorems for 1-ATA and its several natural subclasses. 
Unfortunately, we do not yet have a full solution, although several  structural restrictions on 1-ATA do allow such results.

Firstly, we define  timed extensions of monadic second order logic which we call as $\qkmso$ and its sublogics $\qtwomso$, $\qkmlo$ and $\qtwomlo$. These logics are inspired by the logic $\mathsf{Q2MLO}$ (over continuous time) defined by Hirshfeld and Rabinovich \cite{rabin} as well as Hunter \cite{hunter}. 
In $\qkmso$, logic MSO is extended with a metric quantifier block consisting of at most $k$ quantifiers 
resulting in a formula with exactly one free variable. This can be recursively used as an atomic predicate in MSO.  
 A carefully defined syntax gives us a logic which allows only future time properties to be stated.

As our first main result, we define a subclass of 1-clock ATA called 1-clock ATA with loop-free resets (1-ATA-$\lfr$). In these automata,
on any run and for any location $q$, a reset transition leading to $q$ (denoted $x.q$) must occur at most once. Equivalently
there is no cycle involving $x.q$. 
We show that logic $\regmtl$, defined earlier in \cite{mfcs17}, is expressively equivalent to 1-ATA-$\lfr$. Moreover, we also show that  $\qkmso$ is expressively equivalent to $\regmtl$. In proving this, we show a four variable property showing that $\qkmso$ is expressively equivalent to $\qfourmso$.
A variant of this result allows us to characterize a subclass of 1-ATA called partially ordered 1-ATA with the 
star-free fragment of $\regmtl$ (as shown in \cite {mfcs17}) as well as the first order fragment $\qkmlo$. This in turn is expressively equivalent to $\qfourmlo$. 

As our second main result, we introduce a notion of conjunctive-disjunctiveness in 1-clock ATA. Here, the ATA thread is  either in conjunctive mode or in disjunctive mode at a time,  and it can switch modes only on a reset transition. A restricted version of $\regm$ modality called $\F\regm$ was  defined in
\cite{mfcs17}. A similar modality was also defined by Wilke earlier \cite{Wilke}. 
We show that conjunctive-disjunctive 1-ATA with loop-free resets has exactly the expressive power of $\F\regmtl$ which is the same as $\mathsf{MTL}$ extended with only $\F\regm$ modality. Moreover, $\qtwomso$ is expressively
equivalent to $\fregmtl$. This result extends to other classes of 1-clock ATA. 

However, the case of full 1-clock ATA needs to be investigated. Towards this, we show a B\"uchi-like theorem which says that $\mu\regmtl$, obtained by introducing fixed point operators to $\regmtl$, is expressively equivalent to
full 1-clock ATA. However a Kamp theorem giving a classical logic equivalent to this remains under investigation.

\begin{table}
\begin{tabular}{|c|l|c|c|c|}
  \hline
  Restrictions & Automata & Temporal Logic & Forward Classical Logic & Where\\ \hline
    &  1-ATA & $\sfmtl$ & $\qkmlo \equiv \qfourmlo$ & \cite{mfcs17}, Theorem \ref{aut-tl-1}\\ \cline{2-5}
  $\po$ &  $\wf$ 1-ATA & $\F\sfmtl$ & $\qtwomlo$ & Theorem \ref{bk2-b1},\ref{bk2-k}\\ \cline{2-5}
   & $\wf$ 1-ATA  \textcolor{red}{np} & $\F\sfmitl$ & $\mathsf{Q2FO}$ \textcolor{red}{np} & Theorem \ref{bk2-b1},\ref{bk2-k}\\ \hline
    Loop   &  1-ATA & $\regmtl$ & $\qkmso\equiv \qfourmso$&Theorem \ref{aut-tl-1}, \ref{thm:bk1-k}\\ \cline{2-5}
  Free  & $\wf$ 1-ATA & $\F\regmtl$ & $\qtwomso$ &Theorem \ref{bk2-b1},\ref{bk2-k}\\ \cline{2-5}
  Resets & $\wf$ 1-ATA \textcolor{red}{np} & $\F\regm\mitl$ & $\mathsf{Q2MSO}$ \textcolor{red}{np}&Theorem \ref{bk2-b1},\ref{bk2-k}\\ \hline
  None   &  1-ATA & $\mu \regmtl$ & Open &Thoerem \ref{aut-tl-fixpoints1}, \ref{aut-tl-fixpoints2}\\ \cline{2-5}
   & $\wf$ 1-ATA & $\mu \fregmtl$ & Open &Theorem \ref{aut-tl-fixpoints1}, \ref{aut-tl-fixpoints2}\\ \cline{2-5}
      & $\wf$ 1-ATA \textcolor{red}{np}  & $\mu \fregmitl$ &Open  &Theorem \ref{aut-tl-fixpoints1}, \ref{aut-tl-fixpoints2}\\ \hline
\end{tabular}
\caption{Summary of results. \textcolor{red}{np} stands for non-punctual guards. A non-punctual guard is one which is not of the form $[a,a]$ for $a \in \mathbb{N}$.}
\end{table}

\section{Preliminaries}
\label{sec:prelims}
  Let $\Sigma$ be a finite set of propositions. A finite timed word over $\Sigma$ is a tuple
$\rho = (\sigma,\tau)$, where   $\sigma$ and $\tau$ are sequences $\sigma_1\sigma_2\ldots\sigma_n$ and  $\tau_1\tau_2\ldots \tau_n$ respectively, with $\sigma_i \in \Gamma=2^{\Sigma} \backslash \emptyset$,  and $\tau_i \in \R_{\geq 0}$
 for $1 \leq i \leq n$. For all $ i \in dom(\rho)$,  we have $\tau_i \le \tau_{i+1}$, where $dom(\rho)$ is the set of positions $\{1,2,\ldots,n\}$ in the timed word. For convenience, we assume $\tau_1=0$. 
 The $\sigma_i$'s can be thought of as labeling positions $i$ in $dom(\rho)$.  
 For example, given $\Sigma=\{a,b,c\}$,  $\rho=(\{a,c\},0)(\{a\},0.7)(\{b\},1.1)$ is a timed word.
$\rho$ is strictly monotonic iff $\tau_i < \tau_{i+1}$ for all $i,i+1 \in dom(\rho)$. 
Otherwise, it is weakly monotonic. 
The set of finite timed words over $\Sigma$ is denoted $T\Sigma^*$. 
Given $\rho=(\sigma,\tau)$ with $\sigma=\sigma_1\dots \sigma_n \in \Gamma^+$,
 $\sigma^{\singl}$ denotes the set of 
all words $w_1w_2 \dots w_n$ where each $w_i \in \sigma_i$. 
$\rho^{\singl}$ consists of all timed words 
 $(\sigma^{\singl}, \tau)$. 
For the $\rho$ as above,
$\rho^{\singl}$ consists of timed words $(\{a\},0)(\{a\},0.7)(\{b\},1.1)$
and $(\{c\},0)(\{a\},0.7)(\{b\},1.1)$. 

\subsection{Temporal Logics}
In this section, we define  preliminaries pertaining to the temporal logics studied in the paper. 
Let $I\nu$ be a set of open, half-open or closed time intervals. 
 The  end points of these intervals are  in $\mathbb{N}\cup \{0,\infty\}$.   For example, 
$[1,3), [2, \infty)$. For a time stamp $\tau {\in} \R_{\geq 0}$ and an interval  $\langle a, b\rangle$, where $\langle$ is left-open or left-closed 
and $\rangle$ is right-open or right-closed,   
$\tau+\langle a, b\rangle$ represents  the interval 
$\langle \tau+a, \tau+b\rangle$.

\noindent{\bf Metric Temporal Logic}($\mtl$). Given a finite alphabet $\Sigma$,  the formulae of logic $\mathsf{MTL}$ are built from $\Sigma$  using boolean connectives and 
time constrained version of the until modality $\until$ as follows: 
$\varphi::=a (\in \Sigma)~|true~|\varphi \wedge \varphi~|~\neg \varphi~|
~\varphi \until_I \varphi$, 
where  $I \in I\nu$.    
\label{point}
For a timed word $\rho=(\sigma, \tau) \in T\Sigma^*$, a position 
$i \in dom(\rho)$, and an $\mathsf{MTL}$ formula $\varphi$, the satisfaction of $\varphi$ at a position $i$ 
of $\rho$ is denoted $\rho, i \models \varphi$, and is defined as follows: (i)
\noindent $\rho, i \models a$  $\leftrightarrow$  $a \in \sigma_{i}$, (ii) $\rho,i  \models \neg \varphi$ $\leftrightarrow$  $\rho,i \nvDash  \varphi$, 
(iii) $\rho,i \models \varphi_{1} \wedge \varphi_{2}$   $\leftrightarrow$ 
$\rho,i \models \varphi_{1}$ 
and $\rho,i\ \models\ \varphi_{2}$, (iv)
$\rho,i\ \models\ \varphi_{1} \until_{I} \varphi_{2}$  $\leftrightarrow$  $\exists j > i$, 
$\rho,j\ \models\ \varphi_{2}, \tau_{j} - \tau_{i} \in I$, and  $\rho,k\ \models\ \varphi_{1}$ $\forall$ $i< k <j$. 
 The language of a $\mathsf{MTL}$ formula $\varphi$ is $L(\varphi)=\{\rho \mid \rho, 1 \models \varphi\}$. 
Two formulae $\varphi$ and $\phi$ are said to be equivalent denoted as $\varphi \equiv \phi$ iff $L(\varphi) = L(\phi)$.
 The subclass of $\mathsf{MTL}$ restricting the intervals $I$ in the until modality 
to non-punctual intervals is denoted  $\mathsf{MITL}$.  
 \begin{theorem}[\cite{Ouaknine05}]
 \label{thm-basic}
	 $\mathsf{MTL}$ satisfiability is decidable over finite timed words and is non-primitive recursive. 
\end{theorem}
\paragraph*{\bf{$\mtl$ with Rational Expressions($\regmtl$)}}
We first recall an  extension of $\mathsf{MTL}$ with rational expressions ($\regmtl$), 
 introduced in \cite{mfcs17}.    
The modalities in $\regmtl$ assert the truth of a rational expression (over subformulae) within a particular time interval with respect to the present point. 
For example, for an interval $I{=}(0,1)$, the $\regm_I$ modality works as follows: 
the formula $\regm_{(0,1)}(\varphi_1.\varphi_2)^+$ when evaluated at a point $i$, asserts the existence of  $2k$ points with time stamps 
$\tau_i < \tau_{i+1} < \tau_{i+2} < \dots < \tau_{i+2k} < \tau_{i}+1$, $k >0$, 
such that $\varphi_1$ evaluates to true at $\tau_{i+2j+1}$, and $\varphi_2$ evaluates to true at 
$\tau_{i+2j+2}$, for all $0 \leq j <k$. 

\noindent{$\regmtl$ \bf{Syntax}:} Formulae of $\regmtl$ are built from a finite alphabet $\Sigma$ as: \\
	$\varphi::=a (\in \Sigma)~|true~|\varphi \wedge \varphi~|~\neg \varphi~|~ \regm_I \re(\Ss)~|
	~\varphi \uregm_{I,\re(\Ss)} \varphi$, 
	where $I \in I\nu$
	and $\Ss$ is a finite set of subformulae of $\varphi$, and  $\re(\Ss)$ is defined as a rational expression over $\Ss$. 
	 $\re(\Ss)::= \epsilon~|~\varphi (\in \Ss)~|~\re(\Ss).\re(\Ss)~|~\re(\Ss)+\re(\Ss)~|~[\re(\Ss)]^*$. Thus, $\regmtl$ 
	 is $\mtl$ extended with modalities  $\ureg$ and  $\reg$. 
An \emph{atomic}  rational expression $\re$ is any well-formed formula $\varphi \in \regmtl$. \\
\noindent{$\regmtl$ {\bf Semantics}:} 
For a timed word $\rho=(\sigma, \tau) \in T\Sigma^*$, a position 
$i \in dom(\rho)$,  a $\regmtl$ formula $\varphi$, and a finite set $\Ss$ of subformulae of $\varphi$, we define the satisfaction of $\varphi$ at a position $i$ 
as follows. For positions $i < j \in dom(\rho)$, let $\mathsf{Seg}(\Ss, i, j)$ denote 
the untimed word over $2^{\mathsf{S}}$ 
obtained by marking  the positions $k \in \{i+1, \dots, j-1\}$ of $\rho$ with 
$\psi \in \mathsf{S}$ iff $\rho,k \models \psi$.  
For  a position $i{\in} dom(\rho)$ and an interval $I$, 
 let $\mathsf{TSeg}(\Ss, I, i)$ denote 
the untimed word over $2^{\mathsf{S}}$
obtained by marking all the positions $k$ 
such that $\tau_k - \tau_i \in I$ 
  of $\rho$ with 
$\psi \in \mathsf{S}$ iff $\rho,k \models \psi$. \\
 \noindent$\bullet$ $\rho,i \models \freg_{I,\re(\Ss)} \varphi$  $\leftrightarrow$  $\exists j {>} i$, 
$\rho,j {\models}\ \varphi, \tau_{j} - \tau_{i} {\in} I$ and,  
$[\mathsf{Seg}(\Ss, i, j)]^{\singl} \cap  
 L(\re(\Ss)) \neq \emptyset$, where $L(\re(\Ss))$ is the language  
of the rational expression $\re$ formed over the set $\Ss$. 
In \cite{mfcs17}, the $\ureg$ modality was used instead of $\freg$; however, both have the same 
expressiveness. Note that $\varphi_1 \ureg_{I, \re(\Ss)} \varphi_2$ is equivalent to 
$\freg_{I, \re'(\Ss \cup \{\varphi_1\})} \varphi_2$ where 
$\re'(\Ss \cup \{\varphi_1\})=\re(\Ss) \cap \varphi_1^*$.  \\
\noindent$\bullet$
$\rho,i \models \reg_I \re$ $\leftrightarrow$ 
$[\mathsf{TSeg}(S, I, i)]^{\singl} \cap L(\re(\Ss)) \neq \emptyset$. 
The language accepted by a $\regmtl$ formula $\varphi$ is given by 
$L(\varphi)=\{\rho \mid \rho, 1 \models \varphi\}$. The subclass of $\regmtl$ using only the $\F\regm$ modality is denoted 
$\fregmtl$. When only non-punctual intervals are used, then it is denoted 
 $\fregmitl$.  Thus $\fregmtl$ is $\mtl+\freg$.\\
\noindent{\bf Remark}. The classical $\varphi_1 \until_I \varphi_2$ modality can be written in $\fregmtl$ as 
$\freg_{I, \varphi_1^*} \varphi_2$. 
Also, it can be shown \cite{mfcs17} that the $\freg$ modality 
can be expressed using the $\reg$ modality. 

\noindent{{\bf Modal depth}:} 
We define the modal depth ($\md$) of a $\regmtl$ formula. 
An atomic $\regmtl$ formula over $\Sigma$ is just a propositional 
logic formula over $\Sigma$ and has modal depth 0. 
A $\regmtl$ formula $\varphi$ over $\Sigma$ 
having a single modality ($\regm$ or $\freg$) 
has modal depth one and has the form 
$\regm_{I}\re$ or $\freg_{I, \re} \psi$ 
where  $\re$ is a regular expression over propositonal logic formulae over $\Sigma$ 
and  $\psi$ is a propositional logic formula over $\Sigma$. 
Inductively, we define modal depth as follows.

\noindent (i) $\md(\varphi \wedge \psi)$ = $\md(\varphi \vee \psi)$ = $\max[\md(\varphi),\md(\psi)]$.
 Similarly, $\md(\neg \varphi) = \md(\varphi)$. \\
(ii) Let $\re$ be a regular expression over the set of subformulae $\Ss = \{\psi_1,\ldots,\psi_k\}$. \\
 $\md(\freg_{I,\re}(\varphi)){=} 1{+} \max(\psi_1,\ldots,\psi_k,\varphi)$. 
 Similarly $\md(\regm_I (\re)) {=} 1{+} \max(\psi_1,\ldots,\psi_k)$.

\begin{example}
Consider the formula $\varphi=a \uregm_{(0,1), (aa)^*} b$. Then $\re{=}(aa)^*$, and the subformulae 
of interest are $\Ss=\{a,b\}$. 
For $\rho{=}(\{a\},0)$ $(\{a,b\},0.3)$ $(\{a,b\},0.7)(\{b\},0.9)$, 
 $\rho, 1 {\models} \varphi$, since $a {\in} \sigma_2$, $\sigma_3$, $b {\in} \sigma_4$,  $\tau_4 {-} \tau_1 {\in} (0,1)$ and 
 $aa \in [\mathsf{Seg}(\{a,b\}, 1, 4)]^{\singl} \cap L((aa)^*)$. 
  $[\mathsf{Seg}(\{a,b\}, 1, 4)]^{\singl}$ consists of  the words $aa, ab, ba, bb$.  
On the other hand, for $\rho=(\{a\},0)(\{a\},0.3)$ $(\{a\},0.5)(\{a\},0.9)$ $(\{b\},0.99)$,  
$\rho, 1 \nvDash \varphi$, since even though $b \in \sigma_5, a \in \sigma_i$ for $i <5$, 
$[\mathsf{Seg}(\{a,b\}, 1, 5)]^{\singl}$ consists of only the word $aaa$ and $aaa \notin L((aa)^*)$.
\end{example}
\begin{example}
Consider the formula $\varphi=\regm_{(0,2)}[(a \uregm_{(0,1),b^+}c).e^+]$. The subformulae of interest 
are $\Ss=\{a \uregm_{(0,1),b^+}c, a, b, c, e\}$.
\begin{enumerate}
\item 
 Let $\rho{=}(\{a\},0)(\{a,b\},0.3)(\{a,b,e\},0.4)(\{a,b,e\},0.9)(\{b,c,e\},0.99)(\{b,e\},1.3)(\{e\},1.9)$. 
 We want to check if $[\mathsf{TSeg}(\Ss, (0,2), 1)]^\singl$ intersects with $L((a \uregm_{(0,1),b^+}c).e^+)$. 
Note that we mark positions 2 to 7 with subformulae 
from $\Ss$. Position 2 is marked  $a \uregm_{(0,1),b^+}c$ since $c \in \sigma_5$, 
$a \in \sigma_2, \dots, \sigma_4$, and $[\mathsf{Seg}(\Ss,2,5)]^{\singl}$ 
contains the word $bb \in L(b^+)$.  Similarly, positions 
 3 till 7 are marked $e$.   Hence, $[\mathsf{TSeg}(\Ss, (0,2), 1)]^\singl$ contains the word $(a \uregm_{(0,1),b^+}c).e.e.e.e.e$
which is in $L((a \uregm_{(0,1),b^+}c).e^+)$. We thus have $\rho, 1 \models \varphi$.
 \item 
 For $\rho'{=}(\{b\},0)(\{a,b\},0.3)(\{a,b\},0.4)(\{a,e\},0.9)(\{b,c,e\},0.99)(\{b,e\},1.3)(\{e\},1.9)$, 
$\rho',1 \nvDash \varphi$ since $[\mathsf{TSeg}(\Ss, (0,2), 1)]^\singl$
 does not contain any word in  $L((a \uregm_{(0,1),b^+}c).e^+)$. Note that  we 
 can neither mark position 2 with $a \uregm_{(0,1),b^+}c$, nor mark 
 the remaining positions continuously with $e$. Even if we add $b$ to $\sigma_4$ 
 or add $e$ to $\sigma_3$ (but not both), we still have non-satisfiability.  
  \end{enumerate}
\end{example}

\subsection{Classical Logics}
In this section, we define all the preliminaries pertaining to classical logics 
needed for the paper. We introduce some quantitative variants of classical logics 
 inspired by the logics in \cite{rabin}.
Let $\rho=(\sigma,\tau)$ be a timed word over a finite alphabet $\Sigma$, as before.
We define a real-time logic \emph{forward} $\qkmso$ (with parameter $k \in \Nat$) which is interpreted over such words.
It includes $MSO[<]$ over words $\sigma$ relativized to specify only future properties. This
is extended with a notion of time constraint formula $\psi(t_i)$. 

Let $t_0,t_1, \dots$ be first order variables and $T_0, T_1, \dots$ the monadic second-order
variables. We have the two sorted logic consisting of MSO formulae $\phi$ and time constrained formulae 
$\psi$. Let $a \in \Sigma$, each $t_i$ range over first order variables and $T_i$ over second order
variables. Each quantified first order variable  in  $\phi$ is relativized to the future of some variable, say $t_0$, 
called anchor variable, giving formuale of $MSO^{t_0}$. The syntax of $\phi \in MSO^{t_0}$ is given by: \\
$ t_p{=}t_q|t_p{<}t_q|Q_a(t_p)|T_j(t_i) \mid \varphi {\wedge} \varphi|{\neg} \varphi|
 {\exists} t'. t'{>}t_0 \varphi|{\exists}  T_i \varphi |{\bf \psi(t_p)}$. \\
Here, $\psi(t_p)$
is a time constrained formula whose syntax and semantics are given little later.

A formula in $MSO^{t_0}$ with first order free variables $t_0,t_1, \ldots t_k$ and second-order free variables $T_1, \ldots, T_m$ and which is  relativized to the future of $t_0$ is denoted $\phi(\downarrow t_0,\ldots t_k,T_1, \ldots, T_m)$. (The $\downarrow$ is only to indicate the
anchor variable. It has no other function.)
The semantics of such formulae is as usual. Given $\rho$, positions $a_0,a_1 \ldots a_k$ in $dom(\rho)$, and sets of positions $A_1, \ldots, A_m$ with $A_i \subseteq dom(\rho)$,  we define 
$\rho,(a_1,\ldots,a_k,A_1,\ldots,A_m)$ ${\models}$ $\phi({\downarrow} t_0 ,t_1, \ldots t_k,T_1, \ldots, T_m)$ inductively, as usual.
For instance, 
\begin{enumerate}
\item 	$(\rho,a_1,\ldots,a_k,A_1,\ldots,A_m) {\models} $ $t_i {<} t_j$ iff $a_i {<} a_j$,  
\item $(\rho,a_1,\ldots,a_k,A_1,\ldots,A_m) {\models} $
 $Q_a(t_i)$ iff $a {\in} \sigma(a_i)$,
 \item $(\rho,a_1,\ldots,a_k,A_1,\ldots,A_m) {\models} $
 $ T_j(t_i)$ iff $a_i {\in} A_j$, and 
 \item $(\rho,a_1,\ldots,a_k,A_1,\ldots,A_m) {\models} $ $\exists t_{k}~ t_0{<}t_{k} {\land} 
 \phi({\downarrow} t_0,t_1, \ldots t_k,T_1, \ldots, T_m)$ iff \\
 $\rho,(a_0,\ldots,a_k,A_1,\ldots,A_m) \models \phi({\downarrow} t_0,t_1, \ldots t_k,,T_1, \ldots, T_m)$ for some $a_{k} \geq a_0$. 
\end{enumerate}
 We omit other cases.
The {\bf time constraint} $\psi(t_0)$ has the form 
$\Qq_1t_1 \Qq_2t_2 \dots \Qq_jt_j ~
\phi(\downarrow t_0,t_1,\ldots t_j)$ 
where $\phi \in MSO^{t_0}$ and ${\bf j < k}$, the parameter of logic $\qkmso$. Each quantifier 
$\Qq_i t_i$ has the form $\texists t_i \in t_0+ I_i$ or
$\tforall t_i \in t_0+ I_i$ for a time interval  $I_i$ as in $\mtl$ formulae.
$\Qq_i$ is called a metric quantifier. The semantics of such a formula is as follows.
$(\rho,a_0) \models \Qq_1t_1 \Qq_2t_2 \dots \Qq_jt_j~  \phi(\downarrow t_0,t_1,\ldots t_j)$ iff for $1\leq i \leq j$,
there exist/for all $a_i$ such that $a_0 \leq a_i$ and $\tau_{a_i} \in \tau_{a_0}+I_i$, we have
$\rho,(a_0,a_1 \ldots a_j) \models  \phi({\downarrow} t_0,t_1,\ldots t_j)$.
Note that each time constraint formula has exactly one free variable.
\begin{example}
Let $\rho$=$(\{a\},0)$ $(\{b\},2.1)$ $(\{a,b\}, 2.75)$ $(\{b\},3.1)$ be a timed word. 
Consider the time constraint  $\psi(x)$ $= ~
\texists y \in x+(2, \infty) \texists z \in x+(3,\infty) (Q_b(y) \wedge Q_b(z))$. 
It can be seen that $\rho,1 {\models} Q_a(x) \wedge \psi(x)$.
\end{example}

\noindent{\bf{Metric Depth}}. The \emph{metric depth} of a formula $\varphi$ denoted ($\md(\varphi)$) gives the nesting depth of time constraint
constructs. It is defined inductively as follows: For atomic formulae $\varphi$,  $\md(\varphi)=0$. All the constructs of $MSO^{t_i}$ do not increase $\md$. For example $\md[\varphi_1 \land \varphi_2]= max(\md[\varphi],\md[\varphi_2])$ and $\md[\exists t. \varphi(t)]$. 
However $\md$ is incremented for each application of metric quantifier block.
$\md[\Qq_1t_1 \Qq_2 t_2 \ldots \Qq_j t_j \phi]$ $=$ $\md[\phi] + 1$.
\begin{example}
The sentence $\forall t_0~\tforall t_1 \in t_0 + (1,2)~ \{Q_a(t_1) {\rightarrow} (\texists t_2 \in t_1 + [1,1]~ Q_b(t_2))\}$
accepts all timed words such that for each $a$ which is at distance (1,2) from some time stamp $t$, 
there is a $b$ at distance 1 from it. 
This sentence has metric depth two.
\label{eg:classical} 
\end{example}
\noindent{\bf Special Cases of $\qkmso$}. The case when $k=2$ gives  logic $\qtwomso$. The absence of second order variables and 
second order quantifiers gives logics $\qkmlo$ and $\qtwomlo$. The formula in example \ref{eg:classical} is a $\qtwomlo$ formula. 
Note that our $\qtwomlo$ is the pointwise counterpart of logic $\mathsf{Q2MLO}$ 
studied in \cite{rabin} in the continuous semantics.

\subsection{1-clock Alternating Timed Automata}
\label{prelims:aut}
Let $\Sigma$ be a finite alphabet and let $\Gamma=2^{\Sigma} \backslash \emptyset$.
A 1-clock ATA or 1-ATA \cite{Ouaknine05} is a 5 tuple $\mathcal{A}=(\Gamma, S, s_0, F, \delta)$, where 
 $S$ is a finite set of locations,   
$s_0 \in S$  is the initial location and $F \subseteq S$ is the set 
of final locations.   Let $x$ denote the clock variable in the 1-clock ATA, and $x \in I$ denotes a clock constraint 
where $I$ is an interval. 
 Let $X$ denote a finite set 
of clock constraints of the form $x \in I$.  The transition function is defined as   
$\delta: S \times \Gamma \rightarrow \Phi(S \cup X)$ where
$\Phi(S \cup X)$ is a set of formulae over $S \cup X$ defined by the grammar  
$\varphi::=\top|\bot|\varphi_1 \wedge \varphi_2|\varphi_1 \vee \varphi_2|s|x \in I|x.\varphi$ where  
$s \in S$, and $x. \varphi$ is a binding construct  resetting clock $x$ to 0. 

Without loss of generality, we assume that all transitions $\delta(s,a)$ are in disjunctive 
normal form $C_1 \vee C_2 \vee \dots \vee C_n$ where each $C_i$ is a conjunction 
of clock constraints and locations $s, x.s$. 
Occurrences of $\top$ in
a $C_i$ can be removed, while 
if some $C_i$ contains 
 $\bot$, that $C_i$ can be removed from $\delta(s,a)$.  

A configuration of a 1-clock ATA is a set consisting of locations along with their clock valuation. 
Given a configuration $\Cc=\{(s, \nu) \mid s \in S, \nu \in \Real_{\geq 0}\}$, 
we denote by $\Cc+t$ the configuration $\{(s, \nu+t) \mid s \in S, \nu+t \in \Real_{\geq 0}\}$
obtained after a time elapse $t$, when $t$ is added to all valuations 
in $\Cc$. $\delta(\Cc,a)$ is the configuration obtained by applying 
$\delta(s,a)$ to each location $s$ such that $(s, \nu) \in \Cc$. 
A run of the 1-clock ATA starts from the initial configuration $\Cc_0=\{(s_0,0)\}$ and has the form 
$\Cc_0 \stackrel{t_0}{\rightarrow}\Cc_0+t_0\stackrel{}{\rightarrow} \Cc_1 \stackrel{t_1}{\rightarrow}
  \Cc_1+t_1 \dots \stackrel{}{\rightarrow} \Cc_m$
 and proceeds with alternating time elapse transitions and  
discrete transitions  reading a symbol from $\Sigma$. A configuration $\Cc$ is accepting iff 
for all $(s, \nu) \in \Cc$, $s \in F$. Note that the empty configuration is also an accepting configuration.
 The language accepted  by a 1-clock ATA $\mathcal{A}$, denoted 
$L(\mathcal{A})$ is the set of all timed words $\rho$ such that starting 
from $\{(s_0,0)\}$, reading $\rho$ leads to an accepting configuration. 

We will define some terms which will be used in sections \ref{bk1}, \ref{bk2}.
 Consider a transition $\delta(s,a)=C_1 \vee \dots \vee C_n$
 in the 1-clock ATA. Each $C_i$ is a conjunction 
 of $x \in I$, locations $p$ and $x.p$. We say that $p$ is \emph{free} 
 in $C_i$ if there is an occurrence of $p$ in $C_i$ and no occurrences of $x.p$ in $C_i$;
  if $C_i$  has an $x.p$, then we say that $p$ is \emph{bound} in $C_i$. 
We say that $p$ is bound in $\delta(s,a)$ if it is bound in some $C_i$.  
\noindent A $\po$-1-ATA is one in which \\
\noindent$\bullet$ there is a  partial order denoted $\prec$ on the locations, such that 
the locations appearing in any transition $\delta(s,a)$ 
are in $\{s\} \cup \downarrow s$ where $\downarrow s=\{p \mid p \prec s\}$. \\
\noindent$\bullet$ $x.s$ does not appear in $\delta(s,a)$ for any $s \in S, a \in \Gamma$. \\
It is known \cite{mfcs17} that $\po$-1-ATA exactly characterize logic 
$\sfmtl$ (this is a subclass of $\regmtl$ where the regular expressions have an equivalent star-free expression).  
\begin{example}
 Consider  the $\po$-1-ATA $\mathcal{A}{=}(\{a,b\}$, $\{t_0,t_1,t_2\}, t_0, \{t_0, t_2\}$,
 $\delta_{\Aa})$ with transitions \\ 
	$\delta_{\Aa}(t_0,\{b\}){=}t_0,$  $\delta_{\Aa}(t_0,\{a\})=(t_0 \wedge x.t_1) \vee t_2,$
	$\delta_{\Aa}(t_1,\{a\}){=}(t_1 \wedge x<1) \vee (x>1)=\delta_{\Aa}(t_1,\{b\}),$ and 
	$\delta_{\Aa}(t_2,\{b\}){=}t_2, \delta_{\Aa}(t_2,\{a\})=\bot$, 
	$\delta_{\Aa}(t, \{a,b\})=\bot$ for $t \in \{t_0, t_1, t_2\}$.       
	The automaton accepts all strings where $\{a,b\}$ does not occur, and every non-last $\{a\}$ 
	has no symbols at distance 1 from it, and has some symbol 
	at distance $>1$ from it. 
	\label{eg1}
\end{example}

\subsection{Useful Tools}
\label{sec:tools}
In this section, we introduce some notations and prove some 
lemmas which will be used several times in the paper. 
Let $\Aa$ be a 1-clock ATA and let  $c_{max}$ 
be the maximum constant used in the transitions. 
The set $reg=\{0, (0,1), \dots,$ $c_{max}$, $(c_{max}, \infty)\}$ 
denotes the set of regions. Given a finite alphabet $\Sigma$, 
with $\Gamma=2^{\Sigma}\backslash \emptyset$, a 
region word is a word over the alphabet  $\Gamma \times reg$
  called the \emph{interval} alphabet. 
  A region word $w=(a_1, I_1)$ $(a_2, I_2)$ $\dots (a_m, I_m)$ 
  is \emph{ good} iff $I_j \leq  I_k \Leftrightarrow j <k$ and $I_1$ is initial region $0$. Here,  $I_j \leq I_k$ 
represents that the  upper bound of $I_j$ is at most the lower bound of $I_k$. 
A timed word $(a_1, \tau_1)\dots (a_m, \tau_m)$ is consistent with 
a region word $(b_1, I_1)(b_2, I_2) \dots (b_n, I_n)$
iff $n=m, a_j=b_j$, and $\tau_j \in I_j$ for all $j$. The set of timed words consistent 
with a good region word $w$ is denoted $\Tt_w$. Likewise, given a timed word 
$\rho$, $reg(\rho)$ represents the good region word $w$ such that 
$\rho \in  \Tt_w$. 

The following lemmas (proof of Lemma \ref{untime1} in Appendix \ref{app:untime1}) come in handy later in the paper. 
\begin{lemma}
\label{untime1}[Untiming $P \rightarrow \Aa(P)$]
Let $P$ be a 1-clock ATA over $\Gamma$ having no resets. 
We can construct an alternating finite automaton (AFA)  
$\Aa(P)$ over the interval alphabet $\Gamma \times reg$ 
such that for any good region word $w=(a_1, I_1) \dots, (a_n, I_n)$,  $w \in L(\Aa(P))$ iff  
 $\Tt_w \subseteq L(P)$.   Conversely, $\rho \in L(P)$ iff 
 $reg(\rho) \in L(\Aa(P))$, for any timed word $\rho$.
 Hence, $L(P) ~=~ \{T_w ~\mid~ w \in L(A(P) \}$.
\end{lemma}

\begin{lemma}
\label{untime2}
Let $\Aa(P)$ be an AFA  over the interval alphabet $\Gamma \times reg$ constructed from a reset-free 1-clock ATA $P$ as in Lemma \ref{untime1}. 
 We can construct a $\regmtl$ formula
 $\varphi$ such that $L(\varphi)=\{\Tt_w \mid w \in L(\Aa(P))\}$. Hence, by Lemma \ref{untime1}, $L(\varphi)=L(P)$.
 If the AFA is aperiodic, then $\varphi$ is a $\sfmtl$ formula. 
 \end{lemma}
 \begin{proof}
Let $Det\Aa(P)$ be the deterministic automaton which is language equivalent to $\Aa(P)$.
For any pair of states $p,q$ of $Det\Aa(P)$ and an interval (region) $I_i$,
we can construct a regular expression $\re(p,q,I_i)$ denoting the language
$\{w \in (\Gamma\times \{I_i\})^+ ~\mid~ \delta(p,w)=q \}$. Here $\delta(p,w)$ is the
transition function of $Det\Aa(P)$  extended to words. To  construct $\re(p,q,I_i)$,
let $Det[\Aa(P)[p,q]]$ be the same DFA as $Det[\Aa(P)]$ except that the initial location is $p$
and set of final locations is $\{q\}$. Let $\Aa(I_i^+)$ denote the DFA accepting arbitrary words
in $(\Gamma\times\{I_i\})^+$. Let $Det[\Aa_i] = Det[\Aa(P)[p,q]] \cap \Aa(I_i^+)$ be the automaton
which starts at $p$, accepts on $q$, and has only the $\Gamma\times\{I_i\}$ edges. 
Let $\re(p,q,I_i)$ be the regular expression denoting the
language of $Det[\Aa_i]$.

Consider $Det[\Aa(P)]$.
Let $s_0$ be its initial location, and $F$ be the set of its final locations.
For any sequence of intervals $\mathsf{iseq}=I_{1} < I_{2} < \dots I_{k}$, where $I_{1}$ is the initial
region $0$ and any sequence of locations $\mathsf{sseq}=q_0,q_1,q_2 \ldots q_k$ such that $q_0=s_0$ and $q_k \in F$, 
 the regular expression $\re(q_0,q_1,I_1)\re(q_1,q_2,I_2) \cdots \re(q_{k-1},q_k, I_k)$ denotes subsets of
 good region words accepted by $Det\Aa(P)$ where the control stays within $I_i$ between locations $q_{i-1}$ and
 ${q_i}$.  Define the $\regmtl$ formula 
 $\phi(\mathsf{iseq},\mathsf{sseq})$ $=$ $\regm_{I_1}\re(q_0,q_1,I_1)$ $\regm_{I_2}\re(q_1,q_2,I_2)$ $\cdots$$ \regm_{I_k}\re(q_{k-1},q_k,I_k) ~\mathsf{Last}$,
 where formula $\mathsf{Last} = \reg_{[0, \infty)} \epsilon$ holds only for the last position of a word.
Then $L(\phi(\mathsf{iseq},\mathsf{sseq})){=} \{ \Tt_w~\mid~ w \in L(\re(q_0,q_1,I_1)\re(q_1,q_2,I_2) \cdots \re(q_{k-1},q_k, I_k))\}$, by construction.

Let $\varphi = \vee_{\mathsf{iseq}} \lor_{\mathsf{sseq}} ~\phi(\mathsf{iseq},\mathsf{sseq})$. Then clearly, $L(\varphi)=\{\Tt_w \mid w \in L(\Aa(P))\}$.
Hence, by Lemma \ref{untime1}, $L(\varphi)=L(P)$.
\noindent Note that if we start with a $\po$ 1-clock ATA $P$ in Lemma \ref{untime1}, then the AFA $\Aa(P)$
obtained is aperiodic.  In that case, the regular expressions above have  
a star-free equivalent, resulting in the $\regmtl$ formula being an $\sfmtl$ formula. 
\qedhere
\end{proof}

\noindent{\bf Expressive Completeness and Equivalence}. Let $F_i$ be a logic or automaton class  i.e. a
collection of formulae or automata describing/accepting finite timed words. For each $\phi \in F_i$ let $L(F_i)$ denote the language of $F_i$. We define $F_1 \subseteq_e F_2$ if for each $\phi \in F_1$ there exists $\psi \in F_2$ such
that $L(\phi)=L(\psi)$. Then, we say that $F_2$ is expressively complete for $F_1$.
We also say that $F_1$ and $F_2$ are expressively equivalent, denoted $F_1 \equiv_e F_2$, iff
$F_1 \subseteq_e F_2$ and $F_2 \subseteq_e F_1$.

\section{A Normal Form for 1-clock ATA}
\label{sec:normal}
In this section, we establish a normal form for 1-clock ATA, which plays 
a crucial role in the rest of the paper. 
Let $\Aa = (\Gamma, S, s_0, F, \delta)$ be a 1-clock ATA. $\Aa$ is said to be in 
normal form iff
\begin{itemize}
    \item The set of locations $S$ is partitioned into two sets $S_r$ and $S_{nr}$. The initial state 
$s_0 \in S_r$.
 \item The locations of $S$ are partitioned 
into $P_1,\ldots,P_k$ satisfying the following: Each $P_i$ has a 
unique header location $s_i^r \in S_r$. Also, $P_i - \{s_i^r\} 
\subseteq S_{nr}$. Moreover, for any transition of $\Aa$ of the 
form $\delta(s,a) = C_1 \vee C_2 \ldots C_k$ with $C_i = x \in I 
\wedge p_1 \wedge \ldots \wedge p_m \wedge x.q_i \wedge \ldots 
\wedge x.q_r$ we have (a) each $q_i\in S_r$, and (b) If $s \in 
P_i$ then each $p_j \in P_i-\{s_i^r\}$. 
    \footnote{In section \ref{bk2}, we 
describe a subclass of 1-clock ATA where the transitions 
are sometimes restricted to be in CNF form. The equivalent restriction 
then is that each free location in the transition should be in 
$P_i-\{s_i^r\}$ while each bound location should be in $S_r$.}
\end{itemize}
Each partition $P_j$  can be thought of as an 
\emph{island} of locations. Each island has a unique header 
(or reset) location $s_i^r$. All transitions from outside into $P_j$ 
occur only to this unique header location, and only with reset of 
clock $x$. Moreover, all
 non-reset transitions stays in the same island until a clock is 
reset, at which point, 
the control extends to the header location 
of same or another 
island (this behaviour can be seen on
each path of the run 
tree). 

\noindent \textbf{Establishing the Normal Form} The main result of this
section is that every 1-clock ATA $\Aa$ can be normalized, 
obtaining
a language equivalent 1-clock ATA $\norm(\Aa)$. The
key idea behind this is to duplicate locations of $\Aa$ such 
that
the conditions of normalization are satisfied. Let the set of 
locations of $\Aa$ be  $S = \{s_1,\ldots,s_n\}$. For each location 
$s_i, 1\le i\le n$, create $n+1$ copies, $s_i^r$ and $s_i^
{nr,j}, 1\le j\le n$. If $s_0$ is the initial location of $\Aa$, 
then intial location of $\norm(\Aa)$ is $s_0^r$. The partition 
$P_i$ in $\norm(\Aa)$ will consist of locations $s_i^r, s_h^
{nr,i}$ for $1 \le h \le n$ and entry into $P_i$ happens through 
$s_i^r$. The superscript $r$  on a location represents 
that all incoming transitions to it are on a clock reset, 
while $nr,j$ represents that all incoming transitions to that 
location are on non-reset and it belongs to island $P_j$. In a 
transition $\delta(s_i,a) = \varphi$, all occurrences of 
locations $x.s_j$are replaced as $x.s_j^r$(leading into $P_j$), 
while occurrence of free locations $s_h$ are replaced as $s_h^{nr,i}$. 
The final locations of $\norm(\Aa)$ are $s_i^r,s_i^{nr,j}$ for $ 1\le j \le n$ whenever $s_i$ is a final location in $\Aa$. 
Appendix  \ref{app:normal}  gives a formal proof for the following straightforward lemma.

\begin{lemma}
$L(\Aa) = L(\norm(\Aa))$. 
\end{lemma}
\noindent{\bf{Remark}} Due to above lemma, we assume without loss of generality, in the rest of the paper, that 1-ATA are in normal form.

 \begin{example}
The 1-clock ATA 
$\Bb=(\{a,b\}, \{s_0,s_1,s_2\}, s_0, \{s_1\},\delta)$ with transitions 
	 $\delta(s_0,b)=x.s_2,\delta(s_0,a)=(s_0 \wedge x.s_1)$, 
	$\delta(s_1,a)=(s_1 \wedge s_0)=\delta(s_1,b)$ and  
	$\delta(s_2,b)=x.s_0,\delta(s_2,a)=(s_2 \wedge x.s_1)$ is not in normal form. 
	 Following the normalization technique, 
	we obtain $\norm(\Bb)$ as follows. 
	$\norm(\Bb)$ has locations $S=\{s^r_i,s^{nr,i}_{j} \mid 0 \leq i, j \leq 2\}$
and final locations $\{s^r_1, s^{nr,0}_{1}, s^{nr,1}_1, s^{nr,2}_1\}$. The transitions $\delta'$ 
are as follows.
$\delta'(s^r_0,b)=x.s^r_2,\delta'(s^r_0,a)=(s^{nr,0}_0 \wedge x.s^r_1)$,
 $\delta'(s^{nr,i}_0,b)=x.s^r_2$,
$\delta'(s^{nr,i}_0,a)=(s^{nr,i}_0 \wedge x.s^r_1)$,$\delta'(s^r_1,a)=(s^{nr,1}_{1} \wedge s^{nr,1}_{0})=\delta(s^r_1,b)$, 
$\delta'(s^{nr,i}_{1},a)=s^{nr,i}_{1} \wedge s^{nr,i}_{0}=\delta(s^{nr,i}_{1},b)$,
$\delta'(s^r_2,b)=x.s^r_0,\delta'(s^r_2,a)=(s^{nr,2}_2 \wedge x.s^r_1)$,
 $\delta'(s^{nr,i}_2,b)=x.s^r_0,\delta'(s^{nr,i}_2,a)=(s^{nr,i}_2 \wedge x.s^r_1)$. (See Figure \ref{normalized-fig}).
It is easy to see that $\norm(\Bb)$ is in normal form: 
	$S_r=\{s_i^r \mid 0 \leq i \leq 2\}, S_{nr}=\{s_i^{nr,j}
	\mid 0 \leq i, j \leq 2\}$, and we have the disjoint sets  
	$P_0=\{s_0^r, s_0^{nr,0}\}, P_1=\{s_1^r,s_0^{nr,1},s_1^{nr,1}\}$
	and $P_2=\{s_2^r, s_2^{nr,2}\}$. See Figure \ref{normalized-fig}. 
\end{example}	
	\begin{figure}[h]
\centerline{
\fbox{\includegraphics[width=0.45\textwidth]{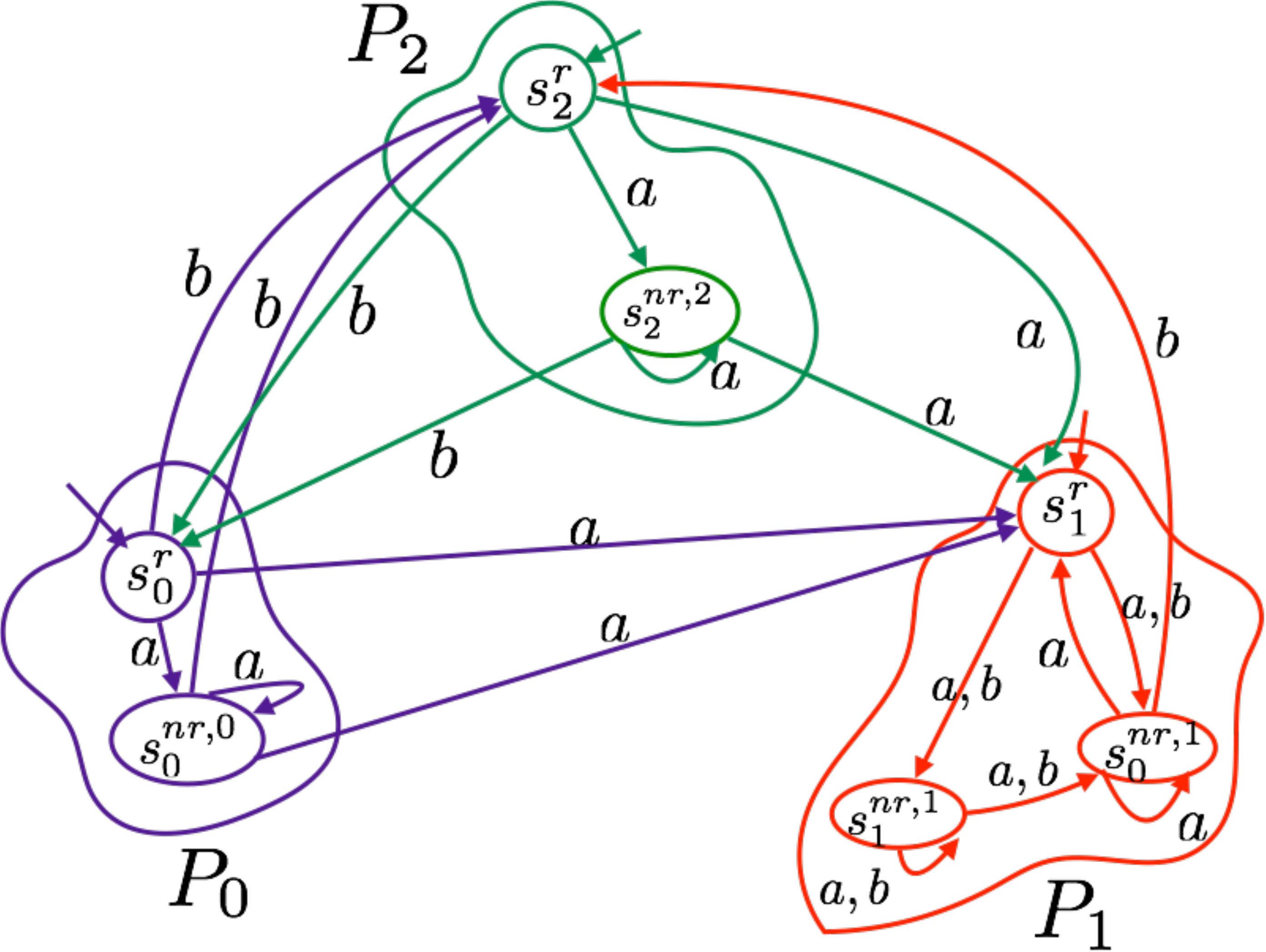}}}
\caption{$\norm(\Bb)$ for $\Bb$}
\label{normalized-fig}
\end{figure}

\section{1-ATA-$\lfr$ and Logics}
\label{bk1}
In this section, we show the first of our B\"uchi-Kamp like results connecting 
logics $\regmtl$, forward $\qkmso$ and a subclass of 1-clock ATA called 
 1-ATA with loop-free resets (1-ATA-$\lfr$).  We first introduce 1-ATA-$\lfr$.

A 1-clock ATA $\Aa$ (in normal form) is said to be a 1-ATA-$\lfr$  if it satisfies the following:
There is a partial order $(S_r,\preceq)$ on the reset states (equivalently, islands $P_i$). Moroeover, for any $p \in P_i$ and location $q$, 
if $x.q$ occurs in $\delta(p,a)$ for any $a$ (giving that $q=s^r_j$) then $s^r_j \prec s^r_i$. Thus,
islands (which are only connected by reset transitions) form a DAG, and
every reset transition goes to a lower level island (see Figure \ref{fig:progr}) where 
this phenomenon is called progressive island hopping.  Semantically, this means that on any branch of
run tree, a reset transition occcurs at most once.
\begin{example}
 The 1-clock ATA with locations $s,p,q$ 
and transitions  
$\delta(s,\alpha){=}(x.p \wedge x \leq 1) \vee (q \wedge x{=}2)$,
$\delta(p,\alpha){=}x.q \wedge p$ and $\delta(q,a){=}s\wedge (0{<} x {<}1)$ 
is not 1-ATA-$\lfr$, since  
$q$ is bound in $\delta(p,\alpha)$ and starting from $q$, we can reach $x.p$ via $s$.
\end{example}
\begin{figure*}[h]
	\includegraphics[scale=0.4]{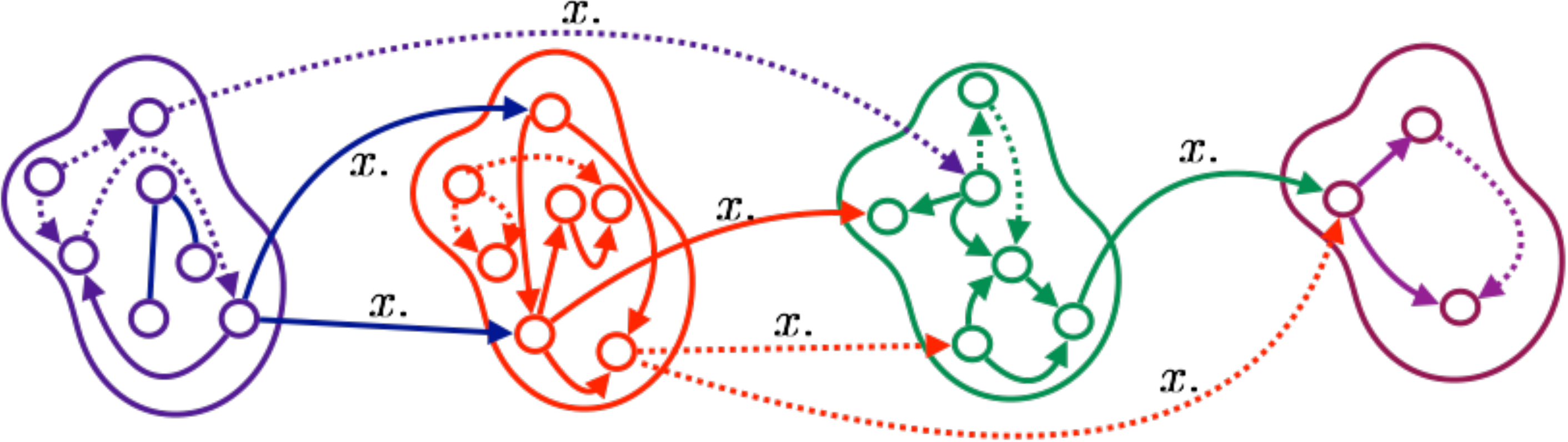}
\caption{A pictorial description of progressive island hopping. The colored islands
are all disjoint.  
On non-reset transitions, control stays in the same island; on resets, it may expand to another island. 
From this finite control, you cannot go back to the island 
where you started from. This provides a partial order between islands due to resets (the name $\lfr$ comes from here) and is referred to as 
 progressive island hopping. For representation purposes, solid arrows indicate conjunctive transitions and dotted ones denote disjunctive transitions.}
\label{fig:progr}
 \end{figure*}

\subsection{B\"uchi Theorem for 1-ATA-$\lfr$}
In this section, we show the equivalence of 1-ATA-$\lfr$  and 
$\regmtl$.  
\begin{theorem}
1-ATA-$\lfr$ are expressively equivalent to $\regmtl$. 
\label{aut-tl-1}
\end{theorem}
When restricting to logic $\sfmtl$, we obtain expressive equivalence with  $\po$-1-ATA.

\subsection{ 1-ATA-$\lfr$ $\subseteq_e$ $\regmtl$}
\label{sec:aut-tl-1}
\proof
Let $\Aa$ be a 1-ATA-$\lfr$ in normal form. For each location $s^r_i \in S_r$ (which is the header
of partition $P_i$) let
$\Aa[s^r_i]$ denote the same automaton as $\Aa$ except that the initial location is changed to $s^r_i$.
We can also delete all islands higher than $P_i$ as their locations are not reachable. 
For each such automaton, we construct a language equivalent
$\regmtl$ formula $mtl(\Aa[s^r_i])$. Note that $(S_r,\preceq)$ is a partial order. The construction
and proof of equivalence are by complete induction on the level of location  $s^r_i$ in the partial order.

Let $s^r_i$ be the header of island $P_i$. All $x.s^r_j$ occuring in 
any transition of $\Aa[s^r_i]$ are of lower level in the partial order $(S_r,\preceq)$. Hence, by induction hypothesis, there is a $\regmtl$ formula $\psi_j=mtl(\Aa[s^r_j])$ language equivalent to $\Aa[s^r_j]$.

Let $w_j$ be a fresh witness variable for each $x.s^r_j$ above, which also corresponds to $\regmtl$
formula $\psi_j$. Let the set of such witness variables be
$\{w_1, \ldots, w_k\}$.   We  construct a modified automaton
$\Aa^{wt}[s^r_i]$ with transition function $\delta'$ and set of locations $P_i$, as follows.
Its alphabet is $\Gamma \times \{0,1\}^k$ with $j$th component giving truth value of $w_j$.
Let $\delta'(s,a,w_1, \ldots, w_k) = \delta(s,a)[w_j/x.s^r_j]$, i.e. in the transition formula
each occurrence of $x.s^r_j$ is replaced by truth value of $w_j$ for $1 \leq j \leq k$.
Note that $\Aa^{wt}[s^r_i]$ is a reset-free 1-ATA. By Lemmas \ref{untime1} and
\ref{untime2}, we get a language equivalent $\regmtl$ formula $\phi^{wt}$ over the variables
$\Sigma \cup \{w_1, \ldots, w_k\}$. Now we substitute each $w_j$ by $\psi_j$ (and hence $\neg w_j$
by $\neg \psi_j$) in $\phi^{wt}$ to obtain the required formula $mtl(\Aa[s^r_i])$. It is clear from
the substitution that $L(\Aa[s^r_i]) = L(mtl(\Aa[s^r_i]))$. \hfill \qed

An example illustrating this construction is in  Appendix \ref{app:eg}. 
 Notice that if $\Aa$ (hence $\norm(\Aa)$) was $\po$-1-ATA, 
then each island $P_i$ is  a  $\po$-1-ATA, and $\Aa(P_i)$ will be an aperiodic automaton;  
hence,  using Lemma \ref{untime2}, we obtain an equivalent $\sfmtl$ formula.
\subsection{$\regmtl$ $\subseteq_e$ 1-ATA-$\lfr$}
\label{thm:logic-aut-1}

\begin{proof}
Consider a formula $\psi_1 = \regm_{I}(\re_0)$ with $I=[l,u)$. The case of other intervals are handled similarly. 
$\psi_1$ has modal depth 1 and has a single modality.   
 As the formula is of modal depth 1, $\re_0$ is an atomic  regular expression over alphabet $\Sigma$.   
Let $D = (\Gamma, Q, q_{0}, Q_f,\delta')$ be a  DFA such that $L(D)=L(\re_0)$, with $\Gamma=2^{\Sigma} \backslash \emptyset$.   
From $D$, we construct the 1-clock ATA  $\Aa{=}(\Gamma,Q \cup \{q_{init}, q_{\mathsf{time check}},q_f\}, q_{init}, \{q_f\}, \delta)$ 
where $q_{init}, q_{\mathsf{time check}},q_f$ are disjoint from $Q$. The transitions 
$\delta$ are as follows. Assume $l >0$. 
\begin{itemize}
	\item $\delta (q_{init},a) = x.q_{\mathsf{time check}}, a \in \Gamma$, 
	\item $\delta (q_{\mathsf{time check}},a) =  [(x \geq  l \wedge \delta'(q_0,a) \vee (q_{\mathsf{time check}})] \vee [x>u \wedge q_f]$
	where the latter disjunct is added only when  $q_0 \in Q_f$,
	\item $\delta (q,a) = (x \in [l,u)) \wedge \delta'(q,a))$, for all $q \in Q \setminus Q_f$, 
	\item $\delta (q,a) = (x\in [l,u) \wedge \delta'(q,a)) \vee (x>u \wedge q_f) $, for all $q \in Q_f$, $\delta(q_f,a) = q_f$. 
\end{itemize}
It is easy to see that $\Aa$ has the loop-free reset condition since $q_f$ is the only location entered 
on resets, and control stays in $q_f$ once it enters $q_f$.  The correctness of $\Aa$ is easy to establish : 
the location $q_{\mathsf{time check}}$ is entered on the first symbol, resetting the clock; 
control stays in $q_{\mathsf{time check}}$ as long as $x<l$, and when 
 $x \geq l$, the DFA is started. As long as $x \in [l,u)$, we simulate the DFA. If $x>u$ and we are in a  final location of the DFA, 
 the control switches to the final location $q_f$ of $\Aa$. 

If $q_0$ is itself a final location of the DFA, then from $q_{\mathsf{time check}}$, we enter 
$q_f$ when $x>u$. It is clear that $\Aa$ indeed checks that $\re_0$ is true in the interval $[l,u)$.
If $l=0$, then the interval on which $\re_0$ should hold good is $[0,u)$. In this case, 
if $q_0$ is non-final, we have the transition $\delta (q_{init},a) = x.\delta'(q_0,a), a \in \Gamma$
(since our timed words start at time stamp 0, the first stmbol is read at time 0, so $x.\delta'(q_0,a)$ preserves the value 
of $x$ after the transition $\delta'(q_0,a)$).  
The location $q_{\mathsf{time check}}$ is not used then.

The case when $\psi_1$ has modal depth 1 but has more than one $\regm$ modality is dealt as follows. 
Firstly, if $\psi_1=\neg \regm_{I}(\re_0)$, then the result follows since  1-ATA-$\lfr$ are closed  under complementation (the fact that 
the resets are loop-free on a run does not change when one complements). For the case when we have a conjunction 
$\psi_1  \wedge \psi_2$ of formulae, having 1-ATA-$\lfr$  $\Aa_1=(\Gamma, Q_1, q_{1}, F_1, \delta_1)$ and 
 $\Aa_2=(\Gamma, Q_2, q_2, F_2, \delta_2)$ such that $L(\Aa_1)=L(\psi_1)$ and 
 $L(\Aa_2)=L(\psi_2)$, we construct  $\Aa=(\Gamma, Q_1 \cup Q_2 \cup \{q_{init}\}, q_{init}, F, \delta)$ 
 such that $\delta(q_{init}, a)=x.\delta_1(q_1,a) \wedge x. \delta_2(q_2,a)$. Clearly, 
 $\Aa$ is a 1-ATA-$\lfr$  since $\Aa_1, \Aa_2$ are. It is easy to see that $L(\Aa)=L(\Aa_1) \cap L(\Aa_2)$.  
 The case when $\psi=\psi_1 \vee \psi_2$ follows from the fact that we handle negation and conjunction.

\paragraph*{Lifting to formulae of higher modal depth}
\label{app:logi-aut-1}
Let us assume the result for formulae of modal depth $\leq k$. Consider a formula of modal depth 
$k+1$ of the form  $\psi_{k+1} = \regm_{I} (\re_k)$, where $\re_k$ is a regular expression over formulae of modal depth $\leq k$. Let  $\psi_k$ be  a formula of modal depth $\leq k$.  
For each such occurrence of a smaller depth formula $\psi_i$, let us allocate a witness variable $Z_i$. 
Let $\Zz=\{Z_1, \dots, Z_k\}$ be the set of all witness variables. 
Given a subset $S \subseteq \Sigma$, let $\Gamma_S \in S \times 2^{\Zz}$.
Any occurrence of an element $S$ in $\re_k$ and  $\psi_k$ are replaced with 
$\Gamma_S$.
 At the end of this replacement, $\re_k$ is a regular expression 
over $\Gamma \times 2^{\Zz}$ and $\psi_k$ is a propositional logic formula over 
$\Gamma \times 2^{\Zz}$. 

Since each $Z_i \in \Zz$ is a witness for a smaller depth formula $\psi_i$, by inductive hypothesis, there is 
a 1-ATA-$\lfr$  $\Aa_{Z_i}$ that is equivalent to $\psi_i$. 
Let $\delta_{Z_i}$ be the  transition function of $\Aa_{Z_i}$ and  let $init_{Z_i}$ 
be the initial location  of $\Aa_{Z_i}$.  
We also construct the complement of each such automata $\Aa_{\neg Z_i}$, which has as its transition function 
$\delta_{\neg Z_i}$ and $init_{\neg Z_i}$ as its initial location. 
The base case  gives us 
a  1-ATA-$\lfr$ (call it $\Cc$)  
over the alphabet  $\Gamma  \times 2^{\Zz}$.  
 Let $\delta_{\Cc}$ denote the transition function of $\Cc$ and let $S_{\Cc}$ be the set of locations 
of $\Cc$. Consider a transition $\delta_{\Cc}(s,\alpha)$ in $\Cc$. 
For $\alpha \in S \times 2^{\Zz}$ the  transition $\delta_{\Cc}(s, \alpha)$ is replaced with 
$\delta'(s, S)$ ${=}$ $\bigvee_{ T \subseteq  \Zz}\delta_{\Cc}(s, \alpha) \wedge$ $
\bigwedge_{\{k \mid Z_k \in T\}} [x.init_{Z_k}]$ $\wedge  \bigwedge_{\{k \mid Z_k \notin T\}} [x.init_{\neg Z_k}]$.
Note that since $\Cc$ as well as $\Aa_{Z_i}$ and $\Aa_{\neg Z_i}$ are 
 1-ATA-$\lfr$ , $\delta'$ also 
respects the $\lfr$ condition. It is easy to see 
that the $\lfr$ condition is respected, since, once we enter the automata $\Aa_{Z_i}$ or $\Aa_{\neg Z_i}$ on reset, we will not return 
to $\Cc$, thereby preserving the linearity of resets.   
Call this 1-clock ATA $\Bb'$.  
Clearly, if $\alpha \in S \times T$ is read in $\Cc$, such that $T{=}\{Z_i, Z_{i_1}, \dots, Z_{i_h}\}$, 
then acceptance in $\Bb'$ is possible iff $\Cc, \Aa_{Z_{i}}$, $\Aa_{Z_{i_1}},$ $\dots, \Aa_{Z_{i_h}}$ 
and $\Aa_{\neg Z_j}$ for $j \neq i, i_1, \dots, i_h$ all reach accepting locations on reading the remaining suffix. 
An example illustrating this can be seen in Appendix \ref{app:logic-aut-eg}.

Note that if $\psi=\reg_{I, \re_0} \psi_0$ is $\sfmtl$, then $\re_0$ is  star-free, 
and all formulae in $\psi_0$ are $\sfmtl$. For the base case, the DFA $D$ obtained will be aperiodic, 
and the 1-clock ATA constructed will be $\po$. The inductive hypothesis guarantees this 
property, and for depth $k+1$, we obtain $\po$ 1-clock ATA since each of
$\delta_{\Cc}, \delta_{Z_k}, \delta_{\neg Z_k}$ satisfy the $\po$ condition, 
and once control shifts to some $\Aa_{Z_i}$, it does not return back to $\Cc$, preserving the $\po$ condition. 
 \end{proof}

\subsection{Kamp Theorem for $\regmtl$ and \emph{forward} $\qkmso$}
In this section, we establish the equivalence between $\regmtl$ 
and \emph{forward} $\qkmso$ giving our first Kamp-like theorem. 
\begin{theorem}
\label{thm:bk1-k}
\emph{forward} $\qkmso$ is expressively equivalent to $\regmtl$. 
\end{theorem}
\oomit{
\begin{lemma}
\emph{forward} $\qkmso \subseteq \regmtl$.
\label{lem:base1}
\end{lemma}
\begin{lemma}
$\regmtl \subseteq$ \emph{forward} $\qkmso$. 
\label{lem:temp-class-1}
 \end{lemma}
}
A careful reading of the proof below also shows that when restricted to $\qkmlo$, we obtain expressive equivalence with logic $\sfmtl$ which is
$\regmtl$ restricted to star-free regular expressions. 
\subsection{\emph{forward} $\qkmso \subseteq_e \regmtl$}
\label{lem:base1}
Proof is by Induction on the metric depth of the formula. For the base case,
consider a formula   $\psi(t_0)$ $=$  $\mathcal{Q}_1 t_1\dots \mathcal{Q}_{k-1}t_{k-1} \varphi(\downarrow t_0, t_1, \dots, t_{k-1})$ of metric depth one.
Let $c_{max}$ be the maximal constant used in the metric quantifiers $\mathcal{Q}_i$. 
Let $R_j(t)$ for $j$ in $reg{=}\{0, (0,1)$,$1,\dots,$  $ c_{max}$, $(c_{max}, \infty)\}$ be fresh monadic predicates. We modify $\psi(t_0)$ to obtain an
untimed MSO formula $\psi_{rg}(t_0)$ over the alphabet $2^{\Sigma} \times \{0,1\}^{|reg|} \times \{0,1\}$ as follows.
Define $CON(I_i,t_i)$ ${=}$  $\lor \{ R_j(t_i) \mid j {\subseteq} I_i\}$.
We replace every quantifier $\texists t_i {\in} t_0+I_i ~\phi$ by $\exists t_i (t_0 {\leq} t_i) {\land} CON(I_i,t_i) {\land} \phi$. 
Every quantifier $\forall t_i {\in} t_0+I_i ~\phi$ is replaced by $\forall t_i .(t_0 {\leq} t_i {\land} CON(I_i,t_i) {\rightarrow} \phi)$.
To the resulting MSO formula we add a conjunct $\mathsf{WELLREGION}$ that states that (a) exactly one $R_j(t)$ holds at any $t$, and (b) 
$\forall t,t'.~ [t {< }t'$ $ {\land} R_j(t) {\land} R_{j'}(t')] ~{\rightarrow}~ j \leq j'$ (asserting region order).  Note that these are natural properties of region abstraction 
of time. This gives us the formula $\psi_{rg}(t_0)$. It has predicates $R_j(t)$ for $j \in reg$ and free variable $t_0$. 
Being MSO formula, we can construct a DFA  $\Aa(\psi_{rg}(0))$ for it over the alphabet  $2^{\Sigma} \times \{0,1\}^{|reg|}$. Note that we have
substituted $0$ for $t_0$.
This is isomorphic to automaton over the alphabet $2^{\Sigma} \times reg$. From the construction, it is clear that $\rho \models \psi(0)$ iff 
$reg(\rho) \in L(\Aa(\psi_{rg}(0)))$.
By Lemma \ref{untime2}, we then obtain an equivalent   $\regmtl$ formula $\zeta$. It is easy to see that $L(\psi(0)) = L(\zeta)$.
Because $\psi(0)$ and $\zeta$ are purely future time formulae, this also gives us that $\rho,i \models \psi(t_0)$ iff $\rho,i \models \zeta$.

For the induction step, consider a metric depth $n+1$ formula $\psi(t_0)$. We can replace every time constraint sub-formula $\psi_i(t_k)$ occurring in it
by a witness monadic predicate $w_i(t_k)$. This gives a metric depth 1 formula  and we can obtain a $\regmtl$ formula, say $\zeta$, over 
variables $\Sigma \cup \{w_i\}$ exactly as in the base step. Notice that each $\psi_i(t_k)$ was a formula of modal depth $n$ or less. Hence by induction hypothesis we have an equivalent $\regmtl$ formula $\zeta_i$. Substituting $\zeta_i$ for $w_i$ in $\zeta$ gives us a formula language equivalent
to $\psi(t_0)$.
\hfill \qed 
 
\subsection{$\regmtl \subseteq_e$ \emph{forward} $\qkmso$} 
\label{lem:temp-class-1}
Let $\varphi \in \regmtl$. The proof is by induction on the modal depth of $\varphi$.
For the base case, let $\varphi=\regm_I(\re)$ where $\re$ is a regular expression over propositions.
Let $\zeta(x,y)$ be an MSO formula with the property that $\sigma,i,j \models \zeta(x,y)$ iff $\sigma[x:y] \in L(re)$, where 
$\sigma[x:y]$ denotes the substring $\sigma(x+1) \dots \sigma(y)$.  
Given that MSO has exactly the expressive power of regular languages, such a formula can always be constructed.
Consider the time constraint formula $\psi(t_0)$:
\[
\begin{array}{l}
 \texists t_{first} \in t_0+I.~\texists t_{last} \in t_0+I.
		\tforall t' \in t_0+I.  {}
  [(t'=t_{first} {\vee} t'=t_{last} 
		{\vee} t_{first} < t' < t_{last}) {\land} 
 ~		\zeta(t_{first},t_{last})]
\end{array}
\]
Then, it is clear that $\rho,i \models \varphi$ iff $\rho,i \models \psi(t_0)$.
Note that $\psi(t_0)$ is actually a formula of $\qkmso$ with $k=4$.

Atomic and boolean constructs can be straightforwardly translated. 
Now let $\varphi=\regm_I(\re)$ where $\re$ is
over a set of subformulae $S$. For each $\zeta_i \in S$, substitute it by a witness proposition $w_i$ to get a formula
$\varphi_{flat}$. This is a modal depth 1 formula and we can construct a language equivalent formula of $\qkmso$, say $\Xi(t_0)$ over alphabet
$\Sigma \cup \{w_i\}$.
By induction hypothesis, for each $\zeta_i$ there exists a language equivalent time constrained $\qkmso$ formula 
$\kappa_i(t_0)$. Now substitute $\kappa_i(t_j)$ for each occurrence of $w_i(t_j)$ in $\Xi(t_0)$ to get a formula
$\psi(t_0)$. Then $\psi(t_0)$ is language equivalent to $\varphi$. Also, by suitably reusing the variables,
$\psi(t_0)$ can be constructed to be in $\qkmso$ with $k=4$.
\hfill \qed

\section{$\wf$-1-ATA-$\lfr$ and Logics}
\label{bk2}
In this section, we show the second of our B\"uchi-Kamp like results connecting 
logics $\fregmtl,$ \emph{forward} $\qtwomso$ and a subclass of 1-clock ATA called 
\emph{conjunctive-disjunctive} (abbreviated $\wf$) 1-clock ATA with loop-free resets. 

Let $\Aa=(\Gamma, Q, q_0, F, \delta)$ be a 1-clock ATA.  Let $Q_x=\{x.q \mid q \in Q\}$ and let 
$\Bb(Q_x)::= \mathsf{true}|\mathsf{false}|\alpha \in Q_x|\alpha \wedge \alpha|\alpha \vee \alpha$.   
 $\Aa$ is said to be a $\wf$ 1-clock ATA if 
  \begin{enumerate}
\item $Q$ is partitioned into 
$Q_{\wedge}$ and $Q_{\vee}$,
\item  Let $q \in Q_{\wedge}$. Transitions
$\delta(q,a)$ can be written as $D_1 {\wedge} D_2 {\wedge}  \dots {\wedge} D_m$, where any $D_i$ has one 
the following forms. (i) $D_i=(q' \vee \Bb(Q_x))$ where $q' \in Q_{\wedge}$, (ii) $D_i=x \in I \vee \Bb(Q_x)$,  
or (iii) $D_i=\Bb(Q_x)$. Thus, each $D_i$ has either at most one free location from $Q_{\wedge}$, or a clock constraint 
$x \in I$. 
\item  Let $q \in Q_{\vee}$. Transitions $\delta(q,a)$ can be written as  $C_1 {\vee} C_2 {\vee} \dots {\vee} C_m$ where any $C_i$ 
 has one of the following forms. (i) $C_i=q' \wedge \Bb(Q_x)$,  where $q' \in Q_{\vee}$, 
 (ii) $C_i=x \in I \wedge \Bb(Q_x)$, or (iii) $C_i=\Bb(Q_x)$.  
  Thus each $C_i$ has either at most one free location from $Q_{\vee}$, or a clock constraint 
$x \in I$. 
\end{enumerate}
The name $\wf$ is based on the fact that each island of $\norm(\Aa)$ is either 
conjunctive or disjunctive.  
  A 1-clock ATA which has both conditions of loop-free resets and conjunctive-disjunctiveness is denoted $\wf$-1-ATA-$\lfr$, while  
 one which satisfies the $\po$ and $\wf$ conditions is denoted 
 $\wf$-1-ATA-$\po$.

  \begin{example}
 \label{example-cd}
    We illustrate examples 
 of ATA which are $\wf$ and which are not. 
 \begin{itemize}
 \item[(a)] The automaton  $\Aa$ with $\delta(s,\{a\}){=}(x {=}1) \vee (x.p \wedge x.r)$, 
 $\delta(p,\{a\}){=}x.s  \vee x.q \vee p$, 
 $\delta(q,\{a\}){=}x.r$, $\delta(r,\{a\}){=}x.q \vee r$
 is a $\wf$,  non $\lfr$ 1-clock ATA.  
 \item[(b)] For $a \in \Sigma$, let  $S_a$ and $S_{\neg a}$ denote any set containing $a$ and not containing $a$, respectively. 
 Consider the automaton $\Bb$ with
 transitions $\delta(s_0,S_a){=}s_0 \vee x > 1$,
 $\delta(s_0,S_{\neg a}){=}s_0 \wedge x \leq 1$, 
 where $s_0$ is the only location, which is non-final. The only way to accept a word is by reaching an empty configuration. 
 The $\wf$ condition is violated due to the combination 
 of having a free location and a clock constraint simultaneously 
 in a clause irrespective of $s_0 \in Q_\vee$ or $s_0 \in Q_\wedge$. This accepts the set of all words where the first symbol in $(1, 2)$ has an $a$.
 \item[(c)] Let $S_a$, $S_{\neg a}$ be as above. The automaton $\Bb$ with
 $\delta(s_0,S_a){=}s_0 \vee s_1$, 
 $\delta(s_2,\Gamma)= x \leq 1$, 
 $\delta(s_0,S_{\neg a}){=}s_0 \wedge s_2$, 
 $\delta(s_1,\Gamma){=}x > 1$
 with $s_0$ being initial and none of the locations being final 
satisfies $\lfr$ but violates $\wf$. 
The $\wf$ condition is violated since a clause contains  
  more than one free location irrespective of $s_0,s_1,s_2 \in Q_\vee$ or $s_0,s_1,s_2 \in Q_\wedge$.
 This accepts the language of all words where the last symbol in $(0, 1)$ has an $a$.
\item [(d)] The automaton $\Bb$ with
 $\delta(s_0,\Gamma){=}s_0 \vee s_1$, 
 $\delta(s_1,S_a )= s_2 \wedge s_3$, 
 $\delta(s_1, S_{\neg a})=\bot$,  
$\delta(s_2,\Gamma){=}(x \le 1)$
$\delta(s_3,\Gamma){=}s_4$
$\delta(s_4,\Gamma){=}x > 1$
 with $s_0$ being initial and $s_1$ being final 
 satisfies $\lfr$ but violates $\wf$. The $\wf$ condition is violated since the 
 automata switches between conjunctive and disjunctive locations without any reset. Note that $s_0 \in Q_{\vee}$ while
 $s_1 \in Q_{\wedge}$. This accepts the language of all words where the second last symbol in $(0,1)$ has an $a$.
\end{itemize}
\end{example}

\subsection{B\"uchi Theorem for $\wf$-1-ATA-$\lfr$}
The main result of this section is the expressive equivalence of $\wf$-1-ATA-$\lfr$ and 
$\fregmtl$. 
 \begin{theorem}
\label{bk2-b1}
	$\wf$-1-ATA-$\lfr$  are expressively equivalent to $\fregmtl$.
\end{theorem}

\subsection{$\wf$-1-ATA-$\lfr$ $\subseteq_e$ $\fregmtl$}
\label{aut-tl2}
 The first thing is to convert $\wf$ 1-clock ATA with no resets to 
	$\F\regm$ formula of modal depth 1 as in Lemma \ref{lem:frat}. 
		\begin{lemma}
	Given a $\wf$  1-clock ATA $\Aa$ over $\Sigma$ with no resets,  
	 we can construct a $\F\regm$ formula $\varphi$ 
 such that for any timed word $\rho=(a_1,\tau_1) \dots (a_m,\tau_m)$, $\rho,i \models \varphi$ iff
  $\Aa$ accepts $(a_i,\tau_i) \dots (a_m, \tau_m)$.
\label{lem:frat}	
	\end{lemma}
\begin{proof}
Assuming $q_0 \in Q_{\vee}$, the key idea is to check how a word is accepted. 
The reset-freeness ensures that any transition  $\delta(q,a)=C_1 \vee \dots \vee C_m$ is such that $C_i$ is either a location  
		or a clock constraint $x \in I$. Assume  acceptance happens through an empty configuration via a clock constraint $x \in I_a$, 
		from some location $q$ on an $a$, and $q$ is reachable from $q_0$. Let 
		$\re_{I_a}$ be the regular expression whose language is the set of all such words reaching some $q$, from where 
		acceptance happens via interval $I_a$ on an $a$. 
		The formula $\freg_{I_a, \re_{I_a}}a$ sums up all such words. Disjuncting over all possible intervals and symbols, we have the result. The second case is when a final state $q_f$ is reached from some $q'$ reachable from $q_0$. If $\re_{q_f,a}$ is the regular expression whose language is all words reaching such a $q'$, the formula 
		$\freg_{(0,\infty)}, \re_{q_f,a}(a \wedge \Box \bot)$ sums up all words accepted via $q', a, q_f$.
		The $\Box \bot$  ensures that no further symbols are read, and can be written as $\neg \freg_{[0, \infty),\Sigma^*} \top$. 
   		 Disjuncting over all possible final states $q_f$ and $a \in \Sigma$ gives us the formula. The case when $q_0 \in Q_{\wedge}$ is handled 
		by negating the automaton, obtaining $q_0 \in Q_{\vee}$ and negating the resulting formula. Details in  Appendix \ref{app:lem:frat}.
		\end{proof}	
	
The rest of the proof is very similar to Section \ref{aut-tl-1} and omitted. Note that if we had started 
with a $\wf$-1-ATA-$\po$, then the regular expressions $\re$ in the $\fregmtl$ formula obtained for the base case
has an equivalent star-free expression, since the underlying  automaton is aperiodic. 
For the inductive case with resets and $\po$, we obtain a $\F \sfmtl$ formula since 
plugging in witness variables with a  $\F \sfmtl$ formula again yields a $\F \sfmtl$ formula.

\subsection{$\fregmtl$ $\subseteq_e$ $\wf$-1-ATA-$\lfr$}
	\label{thm:logic-aut2}
This is almost identical to the proof of section \ref{thm:logic-aut-1}, and is provided in Appendix \ref{app:frat-aut} for completeness.	
Finally, notice that $\F\sfmtl$ formulae correspond to  $\wf$-1-ATA-$\po$.  
\subsection{Kamp Theorem for $\fregmtl$ and \emph{forward} $\qtwomso$}
The expressive equivalence of \emph{forward} $\qtwomso$ and $\fregmtl$ is stated in Theorem \ref{bk2-k}.  
If we restrict to  logic \emph{forward} $\qtwomlo$, then we obtain expressive equivalence with respect to 
$\F \sfmtl$.  
\begin{theorem}
\label{bk2-k}
	$\fregmtl$ is expressively equivalent to  \emph{forward} $\qtwomso$. 
\end{theorem}
\subsection{\emph{forward} $\qtwomso$ ($\qtwomlo$) ${\subseteq}_e$ $\F \regmtl$
($\F \sfmtl$)}
\label{q2mso-freg}
We first consider formulae of metric depth one. These have the form
 $\psi(t_0)=\mathcal{Q}_1 t_1\varphi({\downarrow} t_0, t_1)$ 
 and $\varphi({\downarrow} t_0, t_1)$ is an MSO (FO) formula (bound first order variables  
 $t'$ in $\varphi$ only have the comparison $t' {>} t_0$, and there are no free variables 
 other than $t_0, t_1$, and hence no metric comparison exists in $\varphi$). Let 
 $\re_{\varphi}$ be the regular expression equivalent to $\varphi({\downarrow} t_0, t_1)$. 
 The presence of free variables $t_0, t_1$ implies that $\re_{\varphi}$ ie over the alphabet 
 $2^{\Sigma} \times \{0,1\}^2$, where the last two bits are for $t_0, t_1$. 
 As seen in the case of section \ref{lem:base1}, $t_0$ is assigned the first position 
 of $\re_{\varphi}$ since all other variables take up a position to its right. 
 Hence $\re_{\varphi}$ can be rewritten as $(2^{\Sigma},1,0) \re'$. Since $t_1$ is assigned a unique position, 
 there is exactly one occurrence of a symbol of the form $(2^{\Sigma},0,1)$ in $\re'$. 
 Using (Lemma 7, page 16) \cite{EH01}, we can write $\re'$ as  a finite union 
 of disjoint expressions each of the form $\re_{\ell} (\alpha,0,1) \re_{r}$ where 
 $\alpha \in 2^{\Sigma}$, and $\re_{\ell}, \re_r \subseteq [(2^{\Sigma},0,0)]^*$. 
 $\varphi(\downarrow t_0, t_1)$ is thus equivalent to having 
 a symbol $(\alpha,0,1)$ at a time point $t \in t_0+I$, and 
 $(2^{\Sigma},0,1)\re_{\ell}$ holds till $t$, and beyond $t$, $\re_r$ holds.  
 This is captured 
 by the formula $\freg_{I,\re'}[\bigvee_{\alpha \in 2^{\Sigma}}(\alpha,0,1) \wedge \freg_{(0, \infty), \re_r} \Box \bot]$. 
 Here,  $\re'=(2^{\Sigma},0,1)\re_{\ell}$, and the $\Box \bot$ symbolizes the fact that 
 we see $\re_r$ in the latter part after $(\alpha, 0,1)$ and no more symbols after that.  If $\psi(t_0) \in \qtwomlo$, 
 then $\re_{\varphi}$ is a star-free expression, and so are 
  $\re_{\ell}, \re_r$. That gives us a $\F \sfmtl$ formula.

  The inductive case 
  for formulae of higher depth is in Appendix \ref{app:q2mso-freg}.
The case of going from $\fregmtl$ to \emph{forward} $\qtwomso$ is similar to  section \ref{lem:temp-class-1}, and is in Appendix 
\ref{app:lem:freg-q2mso}.

\noindent{\bf Remark} $\wf$-1-ATA-$\lfr$ with 
	non-punctual guards gives 
expressive equivalence with $\fregmitl$. Likewise, 
$\fregmitl$ is expressively equivalent to logic $\qtwomso$ where 
none of the time constraints are punctual. Note that this is the case since the proof does not introduce punctual 
guards if there are none in the starting automaton/logic.

\section{Temporal Logics with FixPoints}
\label{fixp}
In this section, we look at the logics $\regmtl$ and $\fregmtl$ 
enhanced with fix point operators.  
\paragraph*{\bf{$\regmtl$ with fixed points ($\mu\regmtl$)}} 
\noindent{$\mu\regmtl$ \bf{Syntax}:} Formulae of $\mu\regmtl$ are built from a finite alphabet $\Sigma$ and a finite set $\mathcal{Z}$ of recursion variables as: \\
	$\varphi::=a ({\in \Sigma})|true|Z ({\in \mathcal{Z}})|\varphi \wedge \varphi|\reg_I \re(\Ss)|\F \regm_{I, \re(\Ss)}\varphi|$
	$\greg_{I, \re(\Ss)}\varphi|  \mu Z\circ \varphi|\nu Z \circ \varphi$,  
	where $I {\in} I\nu$, and $\mathsf{S}$, $\re(\Ss)$ are as before, and 
	$\greg_{I, \re(\Ss)}\varphi$ is equivalent to $\neg \freg_{I, \re(\Ss)} \neg \varphi$. 
		The subformulae $\Ss$ of $\varphi$ now can contain $\mu\regmtl$ formulae. 
	A $\mu\regmtl$ formula is said to be \emph{sentence} if every recursion variable $Z$ is within the scope of a fix point operator. Otherwise the formula is open and we write it as $\varphi(Z_1,\ldots,Z_i)$  where $Z_1,\ldots, Z_i$ occur freely.

\noindent{$\mu\regmtl$ {\bf Semantics}:} To define the semantics, we first define a 
 $\emph{super structure}$.  A super structure is a timed word over $[\mathcal{P}(\Sigma)-\emptyset] \times \mathcal{P}(\mathcal{Z})$.
 The super structure is labelled with non-empty subsets of $\Sigma$ and with a possibly empty set of recursion variables at each position. 
For a super structure $\rho=((\sigma,Z), \tau)$, a position 
$i \in dom(\rho)$, a $\mu\regmtl$ formula $\varphi$, and a finite set $\Ss$ of sub-formulae of $\varphi$, we define the satisfaction of $\varphi$ at a position $i$ 
as follows. For $Z \in \mathcal{Z}$, $\rho, i \models Z$
iff $Z \in \sigma(i)$. 
We use the notations $\mathsf{Seg}(\rho,\Ss, i, j)$ and $\mathsf{TSeg}(\rho, \Ss, I, i)$ as in the case of $\regmtl$.  
The semantics of formulae which do not involve $\mu$ are as defined earlier. 

 Two super structures $\rho=(\sigma, \tau)$ and $\rho'=(\sigma', \tau')$ \emph{agree except on $Z$} iff 
$w \in \sigma(i)$ iff $w \in \sigma'(i)$ for all $w \neq Z$ and all $i \geq 1$. 
 We say that a super structure $\rho'$ is a fix point of $Z \equiv \varphi(Z)$ with respect to $\rho$ iff $\rho$ and $\rho'$ agree except on $Z$ and $\rho',i \models \varphi$ iff $\rho',i \models Z$. 	The formula $\mu Z \circ \varphi(Z)$ (respectively $\nu Z  \circ \varphi(Z)$) denotes the 
 \emph{least} (respectively \emph{greatest}) fixpoint solution to the equation $Z \equiv \varphi(Z)$.
  The super structure $\rho'$ is a \emph{least fix point} if 
 whenever $\beta$ is also a fixpoint, then for all  $i \geq 1$, $\rho', i \models Z \Rightarrow \beta, i \models Z$. 
   The super structure $\rho'$ is a \emph{greatest fix point} if  whenever $\beta$ is also a fixpoint, then for all $i \geq 1$, $\beta, i \models Z \Rightarrow \rho', i \models Z$. 
 The semantics of fixed point formulae is as follows.\\
 \noindent{\bf{Semantics of Fix Point Formulae}}:  \\
 \noindent $\bullet$ $\rho, i \models \mu Z \circ \varphi(Z)$ iff $\rho'', i \models Z$ where $\rho''$ is a least fix point for $Z \equiv  \varphi(Z)$ with respect to $\rho$.\\
 \noindent $\bullet$  $\rho, i \models \nu Z \circ \varphi(Z)$ iff $\rho'', i \models Z$ where $\rho''$ is a greatest fix point for $Z \equiv  \varphi(Z)$ with respect to $\rho$.
For sentences $\mu Z \circ \varphi(Z)$, the truth value of $\varphi$ 
is determined using timed words $\rho$ (super structures $\rho=(\sigma, \tau)$ such that 
$Z \notin \sigma(i)$ for all $i$).  
 For a sentence $\varphi$,  a timed word $\rho$, 
and $i \geq 1$, we say that $\rho, i \models \mu Z \circ \varphi(Z)$ if there is a  least fix point  
$\beta$ such that $\beta, i \models \varphi$.  
For any sentence $\varphi$, the language $L(\varphi)$ is defined as set of all the timed words $\rho$ such that $\rho,1 \models \varphi$. If we restrict ourselves to using only $\F\regm$, then the resultant logic is 
 $\mu\fregmtl$.

\begin{example}
Let $\varphi=\mu Z.[a \rightarrow \regm_{(0,1)}[(a+b)^*(b  \vee Z)]$. 
\begin{enumerate}
\item 
Let $\rho'=(\{a\},0)(\{b,Z\},0.6)(\{a\}, 0.9)(\{b,Z\},1.7)(\{a\},1.8)$ and \\
 $\rho=(\{a\},0)(\{b\},0.6)(\{a\},0.9)(\{b\},1.7)(\{a\},1.8)$ be super structures. 
 Then $\rho, \rho'$ agree except on $Z$
and $\rho'$ is a least fixed point of $Z\equiv \varphi(Z)$ with respect to $\rho$.  
It can be seen that $\rho, 1 \nvDash  \mu Z.[a \rightarrow \regm_{(0,1)}[(a+b)^*(b  \vee Z)]$ since 
no super structure $\beta$ which agrees with $\rho$ except on $Z$ can be such that 
$\beta, 1 \models Z$.  

\item Let $\rho'=(\{a,Z\},0)(\{b,Z\},0.6)(\{a,Z\},0.9)(\{b,Z\},1.7)$ and \\
 $\rho=(\{a\},0)(\{b\},0.6)(\{a\},0.9)(\{b\},1.7)$ be super structures. 
  Then $\rho, \rho'$ agree except on $Z$
and $\rho'$ is a least fixed point of $Z\equiv \varphi(Z)$ with respect to $\rho$.  
It can be seen that $\rho, 1 \models  \mu Z.[a \rightarrow \regm_{(0,1)}[(a+b)^*(b  \vee Z)]$ since 
$\rho', 1 \models Z$. 
	 	
\end{enumerate}

\end{example}

 \begin{example}
Let $\varphi=\mu Z.[Z \vee \freg_{[0,1),(aa+Z)^+}b]$.   
\begin{enumerate}
\item Let $\rho'=(\{a,Z\},0)(\{a,Z\},0.2)(\{a,Z\},0.7)(\{b,Z\},0.9)(\{a\},1.1)(\{a\},1.3)(\{b\},1.7)$
and \\ $\rho''=(\{a,Z\},0)(\{a\},0.2)(\{a\},0.7)(\{b,Z\},0.9)(\{a\},1.1)(\{a\},1.3)(\{b\},1.7)$
be two\\  super structures. Then both $\rho'$ and $\rho''$ are fixpoints of $Z \equiv \varphi(Z)$ 
with respect to \\
$\rho=(\{a\},0)(\{a\},0.2)(\{a\},0.7)(\{b\},0.9)(\{a\},1.1)(\{a\},1.3)(\{b\},1.7)$. 
It can be seen that $\rho''$ is a least fix point and $\rho'', 1 \models Z$. Hence, 
$\rho,1 \models \varphi$. 
\item The timed word $\rho= (\{a\},0)(\{a\},0.7)(\{b\},0.9)(\{a\},1.1)(\{a\},1.3)(\{b\},1.7)$
is such that $\rho,1  \nvDash \varphi$. Note that there does not exist a least fix point 
$\beta$ that agrees with $\rho$ except $Z$ such that $\beta, 1 \models Z$.  
\end{enumerate}
\end{example}

 Let $\overline{Z}$ denote a tuple of variables from $\Zz$.

\begin{definition}[Guarded Fragment]
We say that a recursion variable $Z$ is guarded in a temporal $\mu$ calculus formulae $\psi(Z,\overline{Z})$ if and only variable $Z$ is within the scope of a strict future modality. Any formulae is a \it{guarded} formulae if and only if in all its subformulae of the form $\mu Z \circ \psi(Z,\overline{Z})$ (or $\nu Z \circ \psi (Z,\overline{Z})$), $Z$ is guarded in $\psi$. 
 
\end{definition}
It can be easily shown that the guarded restriction on temporal $\mu$ calculus formulae does not affect the expressive power of the logic.\footnote{Replace all the un-guarded variables $Z$ associated with $\nu$ as true and those associated with $\mu$ as false. For details refer \cite{igorbrad}.} 
Hence, we consider only guarded formulae.
 A proof of Lemma \ref{unique_fix_point} is in Appendix \ref{app:fix}.
\begin{lemma}
	\label{unique_fix_point}
	Given any guarded formula  $\psi(Z)$, $Z \equiv \psi(Z)$ has a unique solution if the models are finite timed words.
\end{lemma}
As a corollary, over finite timed words, $\mu Z \circ \varphi(Z)$ is equivalent to  $\nu Z\circ \varphi(Z)$, provided that $Z$ is guarded in $\varphi(Z)$.

\begin{definition}[Temporal Equation Systems]
Consider a series of equations \\
$Z_1 {\equiv^{\mu/\nu}} \psi_1(Z_1, \ldots, Z_m); \ldots;$ $ Z_m {\equiv^{\mu/\nu}} \psi_m(Z_1,\ldots,Z_m)$, where $\psi_1,\ldots,\psi_m$ are temporal logic formulae over $\Sigma \cup \{Z_1,\ldots,Z_m\}$ and $Z_i{\equiv^\mu} \psi_i(Z_1,\ldots,Z_m)$ denotes that $Z_i$ is the least fix point solution of $\psi_i$. 
If the $\psi$ are $\regmtl$ or $\F \regmtl$ formulae then we call it as system of $\regmtl$ equations or $\F \regmtl$ equations, respectively. 
\end{definition}
It can be shown (see \cite{igorbrad}, \cite{brad} and Appendix \ref{app:fix} for an example) that any $\mu \regmtl$ and $\mu \F\regmtl$ can be equivalently reduced to their respective system of equations.
By lemma \ref{unique_fix_point} we know that the least and the greatest fix point operators have identical semantics over finite timed words. Hence,  we will consider only $\mu$ operators and will drop the superscript on $\equiv$. Note that if this equation is true, $Z_i$ is a witness for $\psi_i$. 
The rest of the section establishes the expressive equivalence of 1-clock ATA ($\wf$ 1-clock  ATA) with 
logic $\mu \regmtl$ ($\mu \fregmtl$).   
	\begin{theorem}
		\begin{enumerate}
			\item[(a)] Given  a 1-clock ATA $\Aa$,  there is a 
			$\mu\regmtl$ formula $\psi$ s.t. $L(\psi){=}L(\Aa)$.
			\item[(b)] Given a $\wf$ 1-clock ATA $\Aa$, there is  a 
			$\mu\fregmtl$ formula $\psi$ s.t. $L(\psi){=}L(\Aa)$.
		\end{enumerate}
	\label{aut-tl-fixpoints1}
		\end{theorem}
	\noindent{\bf{Proof Sketch}}: 
	(a) For each island  $P_i$ of $\norm(\Aa)$, we eliminate all the outgoing reset transitions using witnesses as done in section \ref{sec:aut-tl-1}, resulting in a reset-free 1-clock ATA, which in turn is converted to $\regmtl$ formulae $\varphi_i$ over the extended alphabet
	consisting of witness variables $w_j$ for island $P_j$.   
	(b) The islands $P_i$ will be either conjunctive or disjunctive resulting in $\F\regm$ formulae $\varphi_i$ as in section
	 \ref{aut-tl2}. 
		Solving the system $w_1 {\equiv} \varphi_1;\ldots;w_k {\equiv} \varphi_k$ of $\regmtl$ equations,  
		(and $\fregmtl$ in case (b)) the set of words accepted by $\Aa$ is given by the solution for $w_1$.
				
		\begin{theorem}
	\begin{enumerate}
		\item[(a)] Given a  $\mu\regmtl$ formula $\psi$, we can construct a 1-clock ATA $\Aa$ s.t. $L(\psi){=}L(\Aa)$.
		\item[(b)] Given  a $\mu\F\regmtl$ formula $\psi$, we can construct a $\wf$, 1-clock ATA $\Aa$ s.t. $L(\psi){=}L(\Aa)$.
	\end{enumerate}
	\label{aut-tl-fixpoints2}
\end{theorem}
	\noindent{\bf{Proof Sketch}}: 
 The proof is a generalization of sections \ref{thm:logic-aut-1}, \ref{thm:logic-aut2}.
  Given any $\mu\regmtl$ or $\mu\F\regmtl$ formula
  $\varphi$, we can convert it into a system of $\regmtl$ or $\F\regmtl$ equations of the form  
  $Z_1 {\equiv} \psi_1(Z_1, \ldots,Z_m);\ldots; Z_m {\equiv} \psi_m(Z_1, \ldots,Z_m)$. In the case of (a), for all  $\psi_i$, we first construct an equivalent 1-clock ATA with loop free resets, $\Aa_{Z_i}$. As each $\Aa_{Z_i}$ is over $2^{\Sigma} \times 2^{Z_1,\ldots,Z_k}$ where each $Z_i$ is a witness of $\Aa_{Z_i}$, we can eliminate $Z_i$ from all $\Aa_{Z_j}$ by adding reset transitions to $\Aa_{Z_i}$ or $\Aa_{\neg Z_i}$ appropriately as shown in section \ref{thm:logic-aut-1}. 
  For (b), we repeat similar construction obtaining $\wf$-1-ATA-$\lfr$ for each $\psi$. 
   The only difference in (b) is to ensure that after eliminating witnesses, we retain the conjunctive-disjunctiveness of the automata. 
   The formulae $\psi_i$ and $\psi_j$ can depend on each other; $\psi_i$ can contain witness $Z_j$ 
  while $\psi_j$ can contain witness $Z_i$, unlike sections \ref{thm:logic-aut-1}, \ref{thm:logic-aut2}. Due to this circular dependence,  
  while eliminating witnesses, the resulting automaton may not have loop-free resets ($\lfr$).  
   As we need the solution to the first equation, the initial location of the constructed automaton will be the initial location of $\Aa_{Z_1}$.

\begin{theorem}
Satisfiability  of $\mu\fregmitl$ and  reachability in $\wf$ 1-clock ATA with non-punctual guards have elementary decidability.	
\end{theorem}

\begin{proof}
Any $\mu\fregmitl$ formula or $\wf$ 1-clock ATA can be reduced to an equivalent system of $\fregmitl$ equations 
with elementary blow up. 
 	Given any system of equations $Z_1 {\equiv} \psi_1;\ldots;Z_m {\equiv} \psi_m$, the $\fregmitl$  formula
	 $\varphi  {\equiv} Z_1 \wedge \wB$ \footnote{$\wB \varphi$ expands to $\varphi \wedge \Box \varphi$.}
 $(Z_1 \leftrightarrow \psi_1) \wedge 
	 \ldots \wedge \wB(Z_m \leftrightarrow \psi_m) \wedge \wB(\bigvee \limits_{a \in \Sigma} (a))
	 $ over the extended alphabet $2^{\Sigma} \cup 2^{\{Z_1, \dots, Z_m\}}$	   is satisfiable iff 
	 there exists a solution to the above system of equations.     
	 
	 The blow up incurred in the construction of 
	  $\varphi$ is only linear compared to the size of the equations.  
	 	 Note that any $\fregmitl$ formula can be reduced to an $\mitl$ formula 
	 	 	 	  preserving satisfiability with a doubly exponential blow up (elementary) \cite{mfcs17}. Using 
	 	 	 	  the elementary satisfiability \cite{AFH96}
	 	 	 	  	 	 	 	  of $\mitl$, we obtain an elementary upper bound 
	 	 	 	  for  $\mu\fregmitl$. 
	\end{proof}

\section{Discussion} 
We have proposed two new structural restrictions on 1-ATA:
\begin{itemize}
\item[(1)] Loop-Free-Resets, where there are no loops involving
reset transitions, (1-ATA-$\lfr$)
\item[(2)]  Conjunctive-disjunctive partitioning, where the automaton works in purely disjunctive or conjunctive mode between
resets ($\wf$-1-ATA-$\lfr$). 
 Timing constraints only affect resets. In the disjunctive mode, the  automaton behaves like an untimed NFA and the conjunctive mode is its dual.
These structural restrictions were inspired by the quest for automata characterizations of some natural metric temporal logics.
\end{itemize}
 
One of the main contributions in this paper is the study of monadic second order logic with metric quantifiers $\qkmso$ and its subclasses.
We are able to obtain Kamp like theorems with our structural restrictions. It is interesting that
we are able to prove a 4-variable property for $\qkmso$ and $\qkmlo$. It is also  noteworthy that conjunctive-disjunctive restriction
on 1-ATA  uniformly bring the expressiveness down to the two variable fragment.
\oomit{
Our proposed classical logic $\qkmso$ is a generalization of $\mathsf{Q2MLO}$ defined in \cite{rabinovich}. In fact, by restricting $k$ to be $2$, disallowing second order quantification and restricting to non-punctual timing intervals, we exactly get $\mathsf{Q2MLO}$, which we call as "$\mathsf{Q2FO}$ with np guards". We further restrict the timed formulae in all these logics to strictly reason about future only.
}

Finally we give temporal fixpoint logics $\mu\regmtl$ and $\mu\fregmtl$ to characterize full 1-ATA 
and $\wf$-1-ATA. A proper temporal logic and classical logic characterizing the full 1-ATA is left open.
We believe that $\wf$-1-ATA is strictly less expressive than the full 1-ATA, but a formal proof will appear in the full version of this work. The proof goes by extending EF games for $\mtl$ with threshold counting \cite{fossacs16} to that for $\regmtl$.

One of the  takeaways of this paper is the fact that both $\wf$ 1-ATA and $\mu\fregmtl$ enjoy the benefits of relaxing punctuality. That is, the reachability  for $\wf$-1-ATA and satisfiability checking for $\mu\fregmtl$ restricted to non punctual timing constraints are decidable with  elementary complexity. We believe this result is important since, $\wf$ 1-ATA, to the best of our knowledge, is the first such class of timed automata which has alternations and yet the reachability  is decidable with elementary complexity.

\noindent{\bf{Related Work}} : 
B\"uchi's Theorem \cite{Buchi} showing expressive equivalence of MSO$[<]$ and DFA, as well as
Kamp's Theorem \cite{Kamp} showing the expressive equivalence of FO$[<]$ and $LTL$ are classical results.
Going on to timed languages and logics, enhancing regular expressions with quantitative timing properties was first done in \cite{TRE}. Timed regular expressions defined  in \cite{TRE} are exactly equivalent to the class of languages definable by non-deterministic timed automata, and hence not closed under negations. 
Adding regular expressions to LTL was done in \cite{psl}, \cite{psl1}, \cite{DL}. Addition of an automaton modality to $\mitl$ was done by Wilke \cite{Wilke}. Wilke's modality is equivalent to our $\freg$ modality but we also allow punctual intervals.
In \cite{ho},  pointwise $\mtl$ with ``earlier'' and ``newer'' modalities were introduced to obtain
 expressive completeness for FO$[<,+1]$ over bounded timed words.
The temporal logics $\regmtl$ and $\fregmtl$ studied in this paper were first defined in \cite{mfcs17} where their decidability was established.

Expressive completeness for timed logics and languages 
aiming at  B\"uchi-Kamp like theorems has been another prominent line of study. In the timed setting, continuous timed logics have 
been  explored more. Hirshfeld and Rabinovich 
\cite{rabin} showed expressive completeness for $\mitl$ and its counting extension with the subclasses $\mathsf{QMLO}$ and $\mathsf{Q2MLO}$ of 
FO[<,+1]. Their definition of $\mathsf{Q2MLO}$ has been adapted
to pointwise setting and generalized to $\qkmso$ in this paper.
Ouaknine, Worrell and Hunter, in their seminal paper \cite{HunterOW13}, showed expressive completeness for $\mtl$ with rational timing constants with FO [<,+1] (over timed signals). In a related work, \cite{hunter} proved that the expressive completeness carries over even by restricting to standard integer timing constants if  $\mtl$ is extended by threshold counting modality. All these expressive completeness results  were for continuous timed logics which are all undecidable.
Our paper focuses  on point-wise semantics and finite timed words. In this context, the notable result by Ouaknine and Worrell was the reduction of $\mtl[\until_I]$ to partially ordered 1-clock ATA \cite{Ouaknine05}. 
Unfortunately, the converse  does not hold and $\mtl[\until_I]$ is expressively weak. 
Going to full 1-clock ATA, Haase \emph {et al} \cite{Haase} 
  extended 1-TPTL with fixpoints, which is a hybrid between first-order logic and temporal logic, featuring variables and quantification in addition to temporal modalities (quoting \cite{HunterOW13}). They established the expressive equivalence of the two.
 Raskin studied second order extensions of $\mathsf{MITL}$ in both continuous and pointwise time \cite{raskin-thesis}.

\bibliographystyle{plain}
\bibliography{papers}

\newpage
\appendix
\centerline {\bf {\Large Appendix}}

\section{Normal Form for 1-clock ATA}
\label{app:normal}
We start defining a homomorphism between ATA.

\paragraph*{\bf{Homomorphism in 1 clock ATA}}    
Let $\mathcal{A}^1 = (\Sigma, S^1, s^1_0, F^1, \delta^1)$ and $\mathcal{A}^2 = (\Sigma, S^2, s^2_0, F^2, \delta^2)$ 
be 1-clock ATA. We say that $\mathcal{A}^2$ is  \emph{homomorphic} 
to $\mathcal{A}^1$ (denoted $\Aa^2=h(\Aa^1)$) if there is a map $h$ from $S^1$ to $S^2$ satisfying the following.\\
(i)  The map preserves respective initial and final locations: $h(s^1_0) = s^2_0$,   
        and for any $f^1 \in F^1$, $h(f^1) = f^2$ iff $f^2 \in F^2$.\\
    (ii) The map $h$ extends in the usual way to transitions. 
   Corresponding to any transition $\delta^1(s,a)=\varphi$ where $\varphi \in \Phi(S^1 \cup X)$, we obtain 
   the transition  $\delta^2(h(s),a)=h(\varphi)$, where  $h(\varphi)$ is obtained by substituting 
   all occurrences of locations $s \in S^1$ in $\varphi$ with $h(s)$.

\begin{example}
\label{eg2}
The 1-clock ATA $\Aa$ in Example \ref{eg1} is homomorphic to the 1-clock ATA  
$\mathcal{B}=(\{a,b\}, \{s_0,s_1,s_2,s_3\}$, $s_0$, $\{s_0, s_2,s_3\}$,$\delta_{\mathcal{B}})$ with transitions 
	$\delta_{\mathcal{B}}(s_0,b)=s_0,\delta_{\mathcal{B}}(s_0,a)=(s_0 \wedge x.s_1) \vee s_2,$
	$\delta_{\mathcal{B}}(s_1,a)=(s_1 \wedge x<1) \vee (x>1)=\delta_{\mathcal{B}}(s_1,b),$ and 
	$\delta_{\mathcal{B}}(s_2,b)=s_3, \delta_{\mathcal{B}}(s_2,a)=\bot$,    
	$\delta_{\mathcal{B}}(s_3,b)=s_2, \delta_{\mathcal{B}}(s_3,a)=\bot$, under the map $h(s_0)=t_0, h(s_1)=t_1, h(s_2)=t_2=h(s_3)$. 
		$\Aa=h(\Bb)$. 
\end{example}

\begin{lemma}
\label{lem:hom}
Let $\Aa$ and $\Bb$ be 1-clock ATA such that $\Bb=h(\Aa)$. Then $L(\Aa)=L(\Bb)$. 	
\end{lemma}
\begin{proof}
Let $\Aa$ and $\Bb$ be 1-clock ATA such that $\Bb=h(\Aa)$. 
Then we know that each location $s_b$ in $\Bb$ is the map of some location $s_a$ in $\Aa$; moreover,
the initial and final locations of $\Aa$ are mapped to 
initial, final locations respectively in $\Bb$. 
Let $w \in L(\Aa)$, and let $s_0^{\Aa}$ be the initial location of $\Aa$. 
Starting from the initial configuration $\Cc_0=\{(s_0^{\Aa},0)\}$, 
there is a run on $w$ in $\Aa$ which ends in an accepting configuration. 
In $\Bb$, we start with $\Dd_0=\{(h(s_0^{\Aa}),0)\}$. Subsequent configuratons obtained are such that  
$\Dd_i=\{(h(s),t) \mid (s,t) \in \Cc_i\}$. If $w$ was accepted in $\Aa$ 
due to $\Cc_n$ being accepting then we also have $\Dd_n$ accepting due to the property of homomorphisms. 
The converse when $w \in L(\Bb)$ is similar, since we can apply the inverse map of $h$ and draw the same conclusion. 
\end{proof}

\paragraph*{\bf{Normalization of 1-clock ATA}}

Next, we show that for every 1-clock ATA $\mathcal{A}$,  there exists 
a 1-clock ATA $\norm(\Aa)$ in normal form such that $\Aa$ is  homomorphic to $\norm(\Aa)$. 
Let $\mathcal{A} = (\Sigma, S, s_0, F, \delta)$ with $S=\{s_0, s_1, \dots, s_k\}$. 
The normalized ATA $\norm(\Aa)=(\Sigma, S', s_0', F', \delta')$ is as follows:

\begin{itemize}
    \item $S' = \{s^r_i \mid s_i \in S\} \cup \{s_i^{nr,j} \mid s_i {\in} S, 0 \leq j \leq k\}$ 
      \item For every $s_i\in S$ and $a \in \Sigma$, if $\delta(s_i,a) = \varphi$,  
       then for all $0 \leq j \leq k$, $\delta'(s_i^{nr,j},a) = \varphi'$ 
      and $\delta'(s^r_{i},a) = \varphi''$ where $\varphi', \varphi''$ are obtained 
      as follows.
      \begin{itemize}
      \item All locations $s_h$ occurring in $\varphi$ without the binding construct $x.$  
      are replaced in $\varphi'$ with $s_{h}^{nr,j}$, and replaced in $\varphi''$ with $s_{h}^{nr,i}$;
      \item All locations $s_h$  occurring in $\varphi$ as $x.s_h$, with the binding construct $x.$  are replaced in $\varphi', \varphi''$ with 
     $x.s^r_{h}$; 
         \end{itemize}
        \item $s'_0 = s^r_{0}$, $F' = \{s^r_i,s_i^{nr,j}|s_i \in F \wedge 0\leq j \leq k\}$
        \end{itemize}

As we will see, the intuition behind the normalization is that we can find 
disjoint sets $P_1, \dots, P_n$ which partition the set of reachable locations where  
 $P_i=\{s_k^{nr,i}, s_i^r \mid 1 \leq i \leq n\}$. The initial location of a partition $P_i$ is $s^r_i$.
 
\begin{lemma}
$\Aa$ is homomorphic to $\norm(\Aa)$, a 1-clock ATA in normal form.	
\label{lem:normal}
\end{lemma}
\begin{lemma}
 It can be seen that $\Aa$ is homomorphic to $\norm(\Aa)$ ($\Aa=h(\norm(\Aa))$) according to the map
        $h(s^r_i)=h(s_i^{nr,j})=s_i$ for all $0 \leq i, j \leq k$. 	
 To see that $\norm(\Aa)$ is in normal form, consider the partition $S^r=\{s_i^r \mid s_i \in S\}$ and 
 $S^{nr}=\{s_i^{nr,j} \mid s_i \in S, 0 \leq j \leq k\}$. Clearly, locations of $S^r$ appear        
        in transitions attached to the binding construct $x$. while locations of $S^{nr}$ always appear in transitions 
        without the binding construct $x$. Further, 
        the set of locations of $\norm(\Aa)$ can be partitioned into $P_0 \cup \dots \cup P_k$ where 
        $P_j=\{s^r_j, s_i^{nr,j} \mid 0 \leq j \leq k\}$.  For any $s \in P_j$, it is clear from the transitions $\delta'(s,a)=\varphi$ that 
        $s' \in P_j$ for any $s'$ in $\varphi$ iff $s'$ is free in $\varphi$.  Hence, the locations 
        of $P_j$ are either $s_j^r$, (the initial location of $P_j$ which is obtained 
        when $s_j^r$ occurs bound in some transition $\delta(s',a)$ with $s' \in P_k$ for some $k \neq j$), 
         or $s_i^{nr,j}$         where $nr$ represents that location $s_i$ is free in the transition. 
\end{lemma}

\begin{example}
The 1-clock ATA in Example \ref{eg1} is in normal form. We have the partition $S_r=\{t_1\}$, 
$S_{nr}=\{t_0, t_2\}$ and also  $S=P_1 \cup P_2$ with $P_1=\{t_0, t_2\}$ and $P_2=\{t_1\}$. 
The automaton in Example \ref{eg2} is also in normal form. 
\end{example}

\section{Proof of Lemma \ref{untime1}}
\label{app:untime1}
First we describe the construction of $\Aa(P).$
Consider any transition $\delta(s,a)=C_1 \vee \dots \vee C_m$ in the 1-clock reset-free ATA $P$.   
Let $C_{k_1}, \dots, C_{k_m}$ be  clauses containing $x \in I$
for some interval $I$. We consider intervals $I$ in the region form $[0,0] (0,1), \dots, (c_{max}, \infty)$, where 
$c_{max}$ is the maximal constant used in $\Aa$. 
We rewrite $\delta(s,a)$ as  $\delta(s,(a,I))=C'_{k_1} \vee \dots \vee C'_{k_m}$ where $C'_{k_j}$ is obtained from $C_{k_j}$ by  removing the conjunct $x \in I$. Depending on the number of intervals $I$ that appear in $\delta(s,a)$, we obtain transitions  $\delta(s,(a,I))$ by suitably combining clauses that share the same interval.  
   
  The above rewrite of transitions, 
 expands the alphabet to $\Gamma \times reg$, where 
 $reg$ is a set of intervals $[0,0], (0,1), \dots,  (c_{max}, \infty)$ and 
 $\Gamma=2^{\Sigma} \backslash \emptyset$.  
    This rewrite results in making $P$  an 
  untimed alternating finite automaton (AFA) 
  over the \emph{interval alphabet} $\Gamma \times reg$.  
  Let $\Aa(P)$ denote the AFA  with initial location $s_0$ (same as the initial location of $P$).

   The language of $P$ ($L(P)$)
consists of all timed words that have a run from initial location $s_0$ to a  final location 
$q$ of $P$. 
Let $w=(a_1,t_1)\dots (a_m,t_m)$ be a timed word which has a run starting at $\Cc_0=\{(s,0)\}$, where 
$s_0$ is the initial location of  $P$, to an accepting 
configuration $\Cc_m$ in $P$. The run of $w$ on $P$ is 
$\Cc_0 \stackrel{t_1}{\rightarrow} \Cc_0+t_1 \stackrel{a_1}{\rightarrow}\Cc_1 \stackrel{t_2-t_1}{\rightarrow}
\dots \stackrel{t_m-t_{m-1}}{\rightarrow}\Cc_{m-1}+(t_m-t_{m-1}) \stackrel{a_m}{\rightarrow}\Cc_m$. 
Let $t_j \in I_j$.
 By construction of $\Aa(P)$, 
each transition $\delta(s,a_j)=(x \in I_j \wedge \psi) \vee C$ of $P$ has been 
translated into $\delta'(s, (a_j, I_j))= \psi$ (wlg we assume that $C$ has no occurrence of $x \in I_j$). 
Let $\Dd_0=\{s_0\}$.
The run in $P$ now translates into the run $\Dd_0 \stackrel{(a_1,I_1)}{\rightarrow}\Dd_1 \stackrel{(a_2,I_2)}{\rightarrow}
\Dd_2 \dots \stackrel{(a_m, I_m)}{\rightarrow}\Dd_m$, where 
$\Dd_0=\{s \mid (s,0) \in \Cc_0\}$, and 
$\Dd_j=\{s \mid (s,t) \in \Cc_j, t \in I_j\}$, $j \geq 1$. 
Since all locations in $\Cc_m$ are accepting, $\Dd_m$ is an accepting configuration in $\Aa(P)$ accepting
$(a_1,I_1)\dots (a_m,I_m)$.  

Conversely, whenever a good word $(a_1,I_1)\dots (a_m,I_m)$ is accepted by $\Aa(P)$, 
we have an accepting run $\Dd_0 \stackrel{(a_1,I_1)}{\rightarrow}\Dd_1 \stackrel{(a_2,I_2)}{\rightarrow}
\Dd_2 \dots \stackrel{(a_m, I_m)}{\rightarrow}\Dd_m$ as above. By construction of $\Aa(P)$, we obtain a run  
$\Cc_0 \stackrel{t_1}{\rightarrow} \Cc_0+t_1 \stackrel{a_1}{\rightarrow}\Cc_1 \stackrel{t_2-t_1}{\rightarrow}
\dots \stackrel{t_m-t_{m-1}}{\rightarrow}\Cc_{m-1}+(t_m-t_{m-1}) \stackrel{a_m}{\rightarrow}\Cc_m$ 
in $P$ on a word $(a_1,t_1)\dots (a_m,t_m)$ with $t_j \in I_j$
for all $j$. Here, $\Cc_0=\{(s,0) \mid s \in \Dd_0\}$, and 
$\Cc_j=\{(s,t) \mid s \in \Dd_j, t \in I_j\}$ for $j \geq 1$.  
All words $(a_1,t_1)\dots (a_m,t_m)$ with $t_j \in I_j$ will be accepted by $P$, since 
$\Cc_m=\{(s,t) \mid s \in \Dd_m, t\in I_m\}$, and all locations $s$ in $\Dd_m$ are final.

\section{An Example Illustrating Theorem \ref{aut-tl-1}}
\label{app:eg}

\begin{example}
We demonstrate the technique on an example. Consider the timed language consisting of all strings where every $a$ has an even number of $b$'s 
at a distance (1,2) from it. Let the alphabet be $\Sigma=\{a,b\}$. In this example, at any time point, exactly one symbol 
of $\Sigma$ is read. 
This language is accepted by the 
$\lfr$  1-clock ATA $\Aa=(\{a,b\}, \{s_0, s_1, s_2\}, \{s_0\}, \{s_0, s_1\}, \delta)$ with transitions 
\begin{enumerate}
\item $\delta(s_0,a)=s_0 \wedge x.s_1$, $\delta(s_0,b)=s_0$, 
\item  $\delta(s_1, a)=s_1$, 
 $\delta(s_1, b)=(s_1 \wedge x \leq 1) \vee (s_2 \wedge x \in (1,2)) \vee x \geq 2$, 
\item  $\delta(s_2,a)=s_2$, $\delta(s_2,b)=(s_1 \wedge x \in (1,2)) \vee s_2$. 
	
\end{enumerate}

 Note that each of the transitions can be easily made complete with respect to the 
 clock constraints : that is, from each location, 
 for each symbol $a \in \Sigma$ and each interval \\ $I \in \{[0,0], (0,1), [1,1], (1,2), [2,2], (2, \infty)\}$, we 
 have a transition. For instance, the transition $\delta(s_2,b)$ can be easily completed as 
  $(s_1 \wedge x \in (1,2)) \vee (s_2 \wedge x \leq 1) \vee (s_2 \wedge x \geq 2)$.

As a first step, we obtain $\norm(\Aa)$, the normalized form of $\Aa$. \\
$\norm(\Aa)=(\{a,b\}, \{s^r_0, s^{nr,0}_0, s^r_1, s^{nr,1}_1, s^{nr,1}_2\}, \{s_0^r\}, \{s_0^r, s_0^{nr,0}, s_1^r, s_1^{nr,1}\}, \delta')$  
with 
\begin{enumerate}
\item $\delta'(s^r_0,a)=s^{nr,0}_0 \wedge x.s^r_1$, $\delta'(s^r_0,b)=s^{nr,0}_0$,
\item $\delta'(s^{nr,0}_0,a)=s^{nr,0}_0 \wedge x.s^r_1$, $\delta'(s^{nr,0}_0,b)=s^{nr,0}_0$,
\item  $\delta'(s^r_1, a)=s^{nr,1}_1$,  $\delta'(s^r_1, b)=(s^{nr,1}_1 \wedge x \leq 1) \vee (s^{nr,1}_2 \wedge x \in (1,2)) \vee x \geq 2$, 
\item $\delta'(s^{nr,1}_1, a)=s^{nr,1}_1$,  $\delta'(s^{nr,1}_1, b)=(s^{nr,1}_1 \wedge x \leq 1) \vee (s^{nr,1}_2 \wedge x \in (1,2)) \vee x \geq 2$,
\item  $\delta'(s^{nr,1}_2,a)=s^{nr,1}_2$, $\delta'(s^{nr,1}_2,b)=(s^{nr,1}_1 \wedge x \in (1,2)) \vee s^{nr,1}_2$. 
	\end{enumerate}
It can be seen that there are disjoint sets $P_0=\{s^r_0, s^{nr,0}_0\}$,
$P_1=\{s^r_1, s^{nr,1}_1, s^{nr,1}_2\}$, and once a transition 
leaves $P_0$ and enters $P_1$, then it cannot come back to $P_0$. 
$P_1$ is thus a tail island. 
 We rewrite the transitions as follows.

Recall that we expand the alphabet $\Sigma$ to $\Sigma \times \{P_1,\emptyset\} \times \Ii$ 
for $\Aa(P_0)$ and to $\Sigma \times \Ii$ for $\Aa(P_1)$. 
 $\Aa(P_i)$ represents the automaton consisting 
of locations of $P_i$, and transitions between locations of $P_i$. 
All transitions between locations of $\Aa(P_0), \Aa(P_1)$ are reset-free, and hence, 
$\Aa(P_0)$ is an untimed alternating automaton over the alphabet $\Sigma \times \{P_1, \emptyset\} \times \Ii$, 
while $\Aa(P_1)$ is an untimed alternating automaton over the alphabet $\Sigma   \times \Ii$.  
The initial location of $\Aa(P_0)$ is $s_0^r$, while the initial location 
of $\Aa(P_1)$ is $s_1^r$. 
The final locations of  $\Aa(P_0)$ are 
$s_0^r, s_0^{nr,0}$, while the final locations of 
$\Aa(P_1)$ are $s_1^r, s_1^{nr,1}$.

\noindent The transitions in $\Aa(P_0)$ are
\begin{enumerate}
\item $\delta'(s^r_0,(a, P_1,[0, \infty)))=s^{nr,0}_0$, $\delta'(s^r_0,(b, \emptyset, [0, \infty)))=s^{nr,0}_0$,
\item $\delta'(s^{nr,0}_0,(a, P_1,[0, \infty)))=s^{nr,0}_0$, $\delta'(s^{nr,0}_0,(b, \emptyset, [0, \infty)))=s^{nr,0}_0$. 
\end{enumerate}
Clearly, the language accepted by 
$\Aa(P_0)$ is 
$$((a,P_1,[0,1]+(b,\emptyset,[0,1]))^*(a,P_1,(1,2)+(b,\emptyset,(1,2)))^*(a,P_1,[2, \infty))^*+(b, \emptyset,[2, \infty)))^*$$

\noindent The transitions in $\Aa(P_1)$ are 
\begin{enumerate}
\item  $\delta'(s^r_1, (a, [0, \infty)))=s^{nr,1}_1$,  $\delta'(s^r_1, (b, [0,1]))=s^{nr,1}_1$, 
\item $\delta'(s^r_1, (b,  (1,2)))=s^{nr,1}_2$, $\delta'(s^r_1, (b,  [2,\infty))) =\top$,
\item $\delta'(s^{nr,1}_1, (a,  [0, \infty)))=s^{nr,1}_1$,  $\delta'(s^{nr,1}_1, (b, [0,1]))=s^{nr,1}_1$,
\item $\delta'(s^{nr,1}_1, (b,  (1,2)))=s^{nr,1}_2$, $\delta'(s^{nr,1}_1, (b,  [2,\infty))) =\top$,
\item  $\delta'(s^{nr,1}_2,(a,  [0, \infty)))=s^{nr,1}_2$, $\delta'(s^{nr,1}_2,(b,(1,2)))=s^{nr,1}_1$,
\item $\delta'(s^{nr,1}_2,(b,[0,1]))=s^{nr,1}_2$, $\delta'(s^{nr,1}_2,(b,[2,\infty)))=s^{nr,1}_2$.  
	\end{enumerate}

The language accepted by $\Aa(P_1)$ is 
$$(a+b,[0,1])^*+ (a+b,[0,1])^*[(ba^*b,(1,2))]^*+
(a+b,[0,1])^*[(ba^*b,(1,2))]^*[(b+a,[2,\infty))]^*$$
 We obtain the $\regmtl$ formula $\psi_1$  for  $\Aa(P_1)$  as 
$$\regm_{([0,1]}(a+b)^* \vee \regm_{([0,1]}(a+b)^*\regm_{(1,2)}[ba^*b]^*
\vee \regm_{([0,1]}(a+b)^*\regm_{(1,2)}[ba^*b]^*\regm_{[2,\infty)}(a+b)^*$$ while the formula 
$\psi_0$ for $\Aa(P_0)$ is 
$$\regm_{[0,1]}((a,P_1)+b)^*\vee  \regm_{[0,1]}((a,P_1)+b)^*
\regm_{(1,2)}(((a,P_1)+b)^* \vee$$
$$ \regm_{[0,1]}((a,P_1)+b)^*
\regm_{(1,2)}(((a,P_1)+b)^* \regm_{[2,\infty)}(((a,P_1)+b)^*$$

To obtain the correct formula $\psi_0$, we 
replace the symbols $(a,P_1)$ with $a \wedge \psi_1$, giving the formula 
$$\regm_{[0,1]}(a \wedge \psi_1+b)^*\vee  \regm_{[0,1]}(a \wedge \psi_1+b)^*\regm_{(1,2)}((a \wedge \psi_1+b)^* \vee$$
$$\regm_{[0,1]}(a \wedge \psi_1+b)^*\regm_{(1,2)}((a \wedge \psi_1)+b)^* \regm_{[2,\infty)}((a \wedge \psi_1+b)^*$$

Note that since every $a$ is conjuncted with $\psi_1$, the subformula $\regm_{(1,2)}[ba^*b]^*$ of $\psi_1$ ensures that 
if there are non-zero $b$'s at distance (1,2) from $a$, it will be even. 
 \end{example}

\subsection{Example Illustrating Section \ref{thm:logic-aut-1}}
\label{app:logic-aut-eg}
Let $\psi_{3} = \regm_{(0,1)} (d.\regm_{(0,1)}[(a+ \reg_{[1,1]}b^+)^*])$ over $\Sigma=\{a,b,c\}$. 
We rewrite $\psi_3$ as 
$\psi_3=\regm_{(0,1)} (d.\psi_2)$, where 
$\psi_2=\regm_{(0,1)}(a+\psi_1)^*$,  and $\psi_1=\reg_{[1,1]}b^+$. 
   We have 3 formulae here $\psi_1, \psi_2, \psi_3$, and hence 
 3 witness variables $Z_1, Z_2$ and $Z_3$. Let $\Zz_i$ be any subset of 
 $\Zz=\{Z_1, Z_2, Z_3\}$ containing $Z_i$ for $i=1,2,3$. 
Let $\Gamma=2^{\Sigma}\backslash \emptyset \times \Zz$. 
Since we only have single letters from $\Sigma$ true at any point  in the formula 
$\psi_3$, we define 
 for $a \in \Sigma$,  $\Gamma_a \subseteq a \times \Zz$.  
 For each $i=1,2,3$, let $\Gamma_i \subseteq \Sigma \times \Zz_i$. 
 Let $\Sigma_a$ denote any subset of $\Sigma$ containing $a \in \Sigma$.

 \begin{enumerate}
 \item The base case applies to $\psi_1$ and one can construct a $\lfr$ 1-clock 
 ATA $\Aa_{Z_1}$ equivalent to $\psi_1$, and use $Z_1$ as a witness variable for 
 $\psi_1$, wherever it appears. The DFA $D_1$ accepting $\re=b^+$ is as follows.
  $\delta_1(q_0,\Sigma_b)=q_1,  \delta_1(q_1,\Sigma_b)=q_1$, and $q_1$ is a final location
 of $D_1$, and $q_0$ is its initial location. 
 
 The automaton $\Aa_{Z_1}$ has locations $\{q_{init}, q_{\chec}, q_0, q_1, q_f\}$ with 
 $q_f$ as its final location and $q_{init}$ as its initial location. The transitions 
 $\delta_{Z_1}$  are as follows. 
 \begin{enumerate}
 \item  $\delta_{Z_1}(q_{init},\alpha)=x. q_{\chec}$, for $\alpha \in \Sigma_a \cup \Sigma_b \cup \Sigma_c$,
 \item $\delta_{Z_1}(q_{\chec}, \Sigma_b)=((x=1) \wedge q_1) \vee q_{\chec}$, 
 \item $\delta_{Z_1}(q_{\chec}, \alpha)= x < 1 \wedge q_{\chec}$, for $\alpha \in \Sigma_a \cup \Sigma_b \cup \Sigma_c$, 
\item  $ \delta_{Z_1}(q_1,\Sigma_b)=(x=1 \wedge q_1) \vee (x>1 \wedge q_f)$
\item $\delta_{Z_1}(q_f, \alpha)=q_f$,  for $\alpha \in \Sigma_a \cup \Sigma_b \cup \Sigma_c$. 
\end{enumerate}

\item    We can rewrite $\psi_2$ as $\regm_{(0,1)}[(\Gamma_a+\Gamma_1)]^*$. 
 This makes $\psi_2$ a formula of depth one over the extended alphabet 
$\Gamma$. 
 The DFA $D_2$ accepting $(\Gamma_a+\Gamma_1)^*$ is as follows. There are two locations, $s_0$
the initial as well as accepting location, and $s_1$, a dead location. 
$\delta_2(s_0, \alpha)=s_0$ for $\alpha \in \{\Gamma_a, \Gamma_1\}$,
$\delta_2(s_0, \alpha)=s_1$ for $\alpha \notin \{\Gamma_a, \Gamma_1\}$, and finally,
$\delta_2(s_1, \alpha)=s_1$ for all $\alpha \in \Gamma$. 
The automaton $\Aa_{Z_2}$  has locations $\{s_0, s_1, s_{init}, s_{\chec}, s_f\}$ 
with $s_f$ as its final location and  $s_{init}$ as its initial location. 
The transitions $\delta_{Z_2}$
are as follows. 
\begin{enumerate}
\item $\delta_{Z_2}(s_{init}, \alpha)=[x. s_{\chec}]$  for all $\alpha \in \Gamma$,
\item $\delta_{Z_2}(s_{\chec},\alpha)=[x \geq 1 \wedge s_f \wedge x.\delta_{Z_1}(q_{init},\alpha)]$, if $\alpha \in \Gamma_1$,
\item $\delta_{Z_2}(s_{\chec},\alpha)=[x \geq 1 \wedge s_f \wedge x.\delta_{\neg Z_1}(q_{init},\alpha)]$, if $\alpha \in \Gamma_a, \alpha \notin \Gamma_1$, 
 \item $\delta_{Z_2}(s_f, \alpha)=s_f \wedge x.\delta_{Z_1}(q_{init},\alpha)$ if $\alpha \in \Gamma_1$,
 \item $\delta_{Z_2}(s_f, \alpha)=s_f \wedge x.\delta_{\neg Z_1}(q_{init},\alpha)$ if $\alpha \notin \Gamma_1$.   
 \end{enumerate}
 Note that wherever we encounter $\Gamma_1$, we parallely start checking 
 $\Aa_{Z_1}$, and at places where we do not have $\Gamma_1$, we 
 start checking  $\Aa_{\neg Z_1}$. 
\item   Replacing $\psi_2$,  we obtain 
$\psi_3=\regm_{(0,1)} (\Gamma_d \Gamma_2)$ as a modal depth 1 formula over the alphabet 
 $\Gamma$.
 This results in a  modal depth 1 formula over the extended alphabet consisting of witness variables. 
We now obtain the 1-clock $\lfr$ ATA for $\psi_3=\regm_{(0,1)} (\Gamma_d \Gamma_2)$.  The DFA $D_3$ accepting $\Gamma_d \Gamma_2$ 
has locations $r_0, r_1,r_2,r_3$ where $r_0$ is initial, $r_2$ 
is accepting and has transitions 
$\delta_3(r_0, \Gamma_d)=r_1, \delta_3(r_1,\Gamma_2)=r_2$ and 
$\delta_3(r_0, \alpha)=\delta(r_1, \beta)=r_3$ for $\alpha \notin \Gamma_d, \beta \notin \Gamma_2$ and $\delta_3(r_2, \alpha)=r_3=\delta_3(r_3, \alpha)$ for all $\alpha$.

The automaton $\Aa_{Z_3}$ has as locations $\{r_0, r_1, r_2, r_{init}, 
r_f, r_{\chec}\}$ with $r_{init}$ as the initial location and $r_f$ as the 
final location. The transitions $\delta_{Z_3}$ 
are as follows.
\begin{enumerate}
 \item  $\delta_{Z_3}(r_{init},\alpha)=x. r_{\chec}$, for $\alpha \in \Gamma$,
  \item $\delta_{Z_3}(r_{\chec}, \alpha)=(x >0 \wedge r_1) \vee (r_{\chec} \wedge x=0)$, for $\alpha \in \Gamma_d$, $\alpha \notin \Gamma_2, \Gamma_1$,
  \item $\delta_{Z_3}(r_{\chec}, \alpha)=[(x >0 \wedge r_1) \vee (r_{\chec} \wedge x=0)]\wedge $\\
 $\{ [x.\delta_{Z_2}(s_{init},\alpha) \wedge x. \delta_{Z_1}(q_{init}, \alpha)]\}, \alpha \in \Gamma_2, \Gamma_1$,
 \item $\delta_{Z_3}(r_{\chec}, \alpha)=[(x >0 \wedge r_1) \vee (r_{\chec} \wedge x=0)]\wedge$
$[x.\delta_{Z_2}(s_{init},\alpha) \wedge x. \delta_{\neg Z_1}(q_{init},\alpha)]$,  $\alpha \in \Gamma_2, \alpha \notin \Gamma_1$,
\item  $\delta_{Z_3}(r_1,\alpha)=(x>0 \wedge r_2 \wedge$ 
      $[x. \delta_{Z_2}(s_{init}, \alpha)) \wedge x. \delta_{Z_1}(q_{init},\alpha)]$, $\alpha \in \Gamma_2, \Gamma_1$,
\item $\delta_{Z_3}(r_1,\alpha){=}(x>0 \wedge r_2 \wedge$ 
       $[x.\delta_{Z_2}(s_{init},\alpha) \wedge x. \delta_{\neg Z_1}(q_{init},\alpha)]$, $\alpha \in \Gamma_2, \alpha \notin \Gamma_1$,     
  \item  $\delta_{Z_3}(r_2,\alpha)=[(0<x<1 \wedge r_3) \vee (x\geq  1 \wedge r_f)] \wedge$ \\
  $~~~~~~~~~~~~~~~~~~[x. \delta_{Z_2}(s_{init}, \alpha)) \wedge x. \delta_{Z_1}(q_{init},\alpha)]$,  $\alpha \in \Gamma_2, \Gamma_1$,
  \item  $\delta_{Z_3}(r_2,\alpha)=[(0<x<1 \wedge r_3) \vee (x\geq  1 \wedge r_f)] \wedge$ \\
  $~~~~~~~~~~~~~~~~~~[x. \delta_{\neg Z_2}(s_{init}, \alpha)) \wedge x. \delta_{Z_1}(q_{init},\alpha)]$,  $\alpha \in \Gamma_1, \alpha \notin \Gamma_2$,
  \item  $\delta_{Z_3}(r_2,\alpha)=[(0<x<1 \wedge r_3) \vee (x\geq  1 \wedge r_f)] \wedge $\\
  $~~~~~~~~~~~~~~~~~~[x. \delta_{Z_2}(s_{init}, \alpha)) \wedge x. \delta_{\neg Z_1}(q_{init},\alpha)]$,  $\alpha \in \Gamma_2, \alpha \notin \Gamma_1$,
  \item   $\delta_{Z_3}(r_2,\alpha)=[(0<x<1 \wedge r_3) \vee (x\geq  1 \wedge r_f)] \wedge$ \\
  $~~~~~~~~~~~~~~~~~~[x. \delta_{\neg Z_2}(s_{init}, \alpha)) \wedge x. \delta_{\neg Z_1}(q_{init},\alpha)]$,  $\alpha \notin \Gamma_2, \alpha \notin \Gamma_1$,
  \item $\delta_{Z_3}(r_f, \alpha)=r_f \wedge 
  [x. \delta_{Z_2}(s_{init}, \alpha)) \wedge x. \delta_{Z_1}(q_{init},\alpha)]$,  $\alpha \in \Gamma_2, \Gamma_1$,
  \item $\delta_{Z_3}(r_f, \alpha)=r_f \wedge 
  [x. \delta_{\neg Z_2}(s_{init}, \alpha)) \wedge x. \delta_{\neg Z_1}(q_{init},\alpha)]$,  $\alpha \notin \Gamma_2, \Gamma_1$,
  \item $\delta_{Z_3}(r_f, \alpha)=r_f \wedge 
  [x. \delta_{\neg Z_2}(s_{init}, \alpha)) \wedge x. \delta_{Z_1}(q_{init},\alpha)]$,  $\alpha \notin \Gamma_2, \alpha \in \Gamma_1$,
  \item $\delta_{Z_3}(r_f, \alpha)=r_f \wedge 
  [x. \delta_{Z_2}(s_{init}, \alpha)) \wedge x. \delta_{\neg Z_1}(q_{init},\alpha)]$,  $\alpha \in \Gamma_2, \alpha \notin \Gamma_1$
  \end{enumerate}
Note that if at a point, the symbol read is in both $\Gamma_2$ and $\Gamma_1$, then we 
start $\Aa_{Z_1}$ and $\Aa_{Z_2}$ in parallel. Likewise, if the symbol read is in $\Gamma_1$ but not in $\Gamma_2$, 
then we start $\Aa_{Z_1}$ and $\neg \Aa_{Z_2}$ in parallel.  \end{enumerate}

 \section{Example Illustrating Section \ref{lem:base1}}
\label{app:qkmso-regmtl}
In this section, we show how to convert a $\qkmso$ formula 
into a $\regmtl$ formula, as described in Section \ref{lem:base1}. 
Let us apply to the formula  $\psi(t_0)$ in Example \ref{eg:classical}. We also show here, how to compute the automaton equivalent to $\psi(t_0)$.  
\begin{example}
\begin{enumerate}
\item Using the extra predicates $R_{(1,2)}(t_1), R_1(t_2)$ we obtain \\
$\forall t_1 \{[Q_a(t_1) \wedge R_{(1,2)}(t_1)] \rightarrow \exists t_2[Q_b(t_2) \wedge R_1(t_2)]\}$. 
Note that we can draw the automaton corresponding to the predicate 
$R_{(1,2)}(t_1)$ as follows. The free variable $t_0$ is assigned the first position. 

\begin{figure}[h]
\includegraphics[scale=0.3]{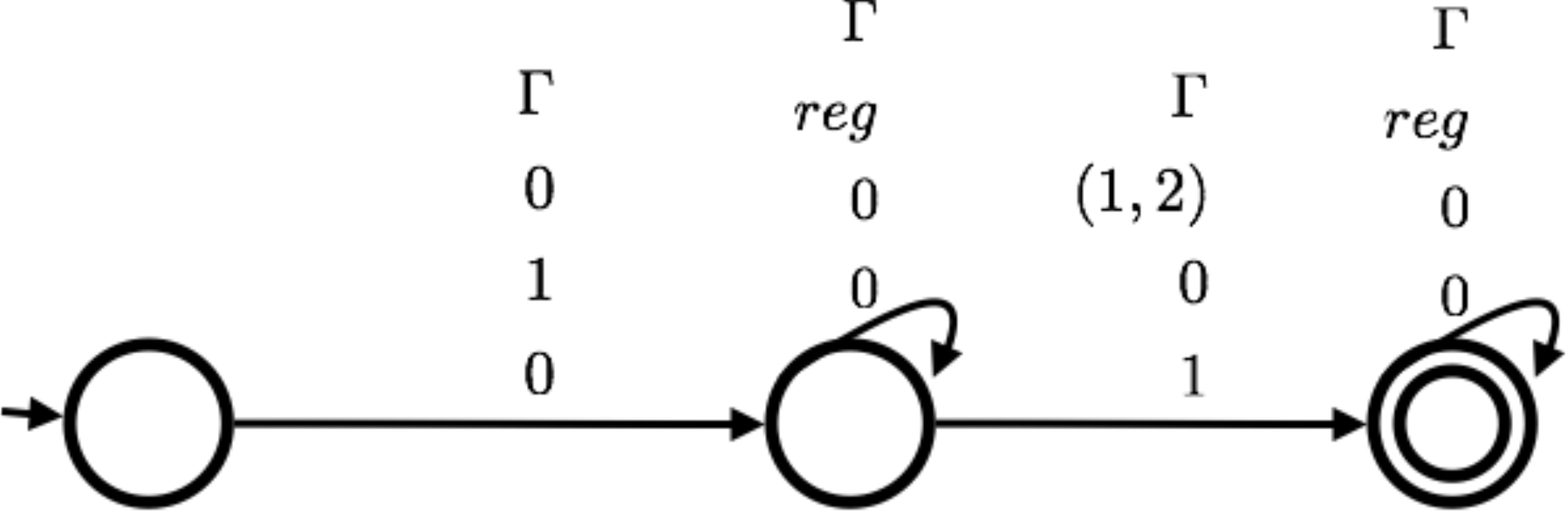}	
\includegraphics[scale=0.25]{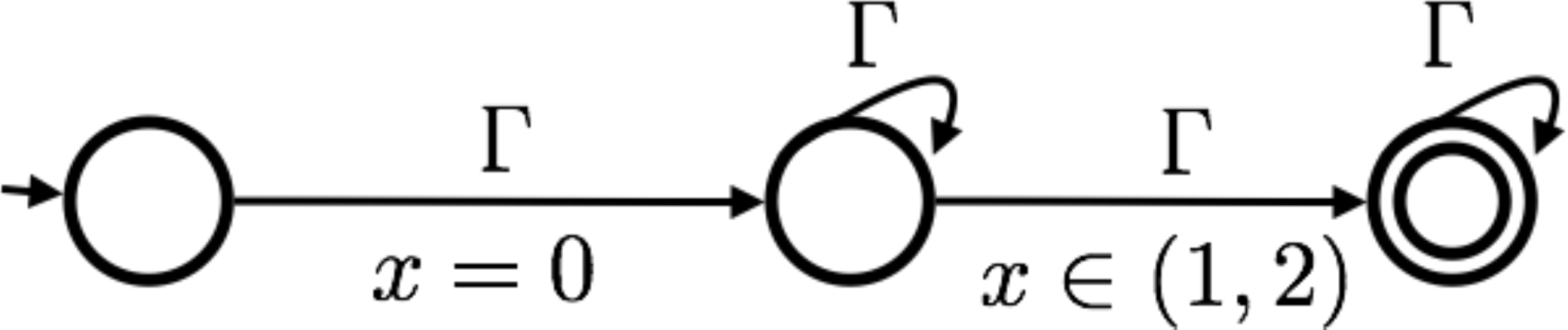}	
\caption{$\Gamma=2^{\Sigma} \backslash \emptyset$. The alphabet for the DFA is $\Gamma \times reg \times \{0,1\}^2$.
$reg$ is the set consisting of 0, (0,1), 1, (1,2), 2 and (2, $\infty)$. 
 On the right, the corresponding timed automata obtained by removing 
 elements $\alpha \in reg$ with the clock constraint $x \in \alpha$. 
 The self loops can have constraints $x \in I$  compatible with $x \in (1,2)$.}
\label{fig:pred-r}
\end{figure}

\item Substituting a witness $W$ for  $\exists t_2[Q_b(t_2) \wedge R_1(t_2)]$ we obtain \\
$\forall t_1 \{[Q_a(t_1) \wedge R_{(1,2)}(t_1)] \rightarrow W]\}$, which is rewritten as 
$\chi=\neg \exists t_1 \{Q_a(t_1) \wedge R_{(1,2)}(t_1) \wedge \neg W]\}$. The timed automata 
for $W, \neg W$ can be seen in Figure  \ref{eg-class1}.  This is obtained by first constructing the 
DFA for $W$, and then putting in the time constraint in the same way we dealt with 
Figure \ref{fig:pred-r}.  
\item The automaton on the top of Figure \ref{eg-class1} is over the 
extended alphabet $2^{{\Sigma} \cup W}$ where $W$ is the witness symbol. 
To get the automaton equivalent to $\psi(t_0)$, we 
replace symbols $W$ and $\neg W$, by replacing the transitions.
From the automaton on the top in Figure \ref{eg-class1}, we obtain the
transition $\delta(s_1, S_a)= [x\in (1,2) \wedge s_1 \wedge x.q_0^{W}] \vee 
[x\in (1,2) \wedge s_2 \wedge x.q_0^{\neg W}]$.  Note that 
each time an $a$ is read in time (1,2), acceptance is possible only when there is a $b$ at distance 1.   
\end{enumerate}

\begin{figure}[h]
	\includegraphics[scale=0.25]{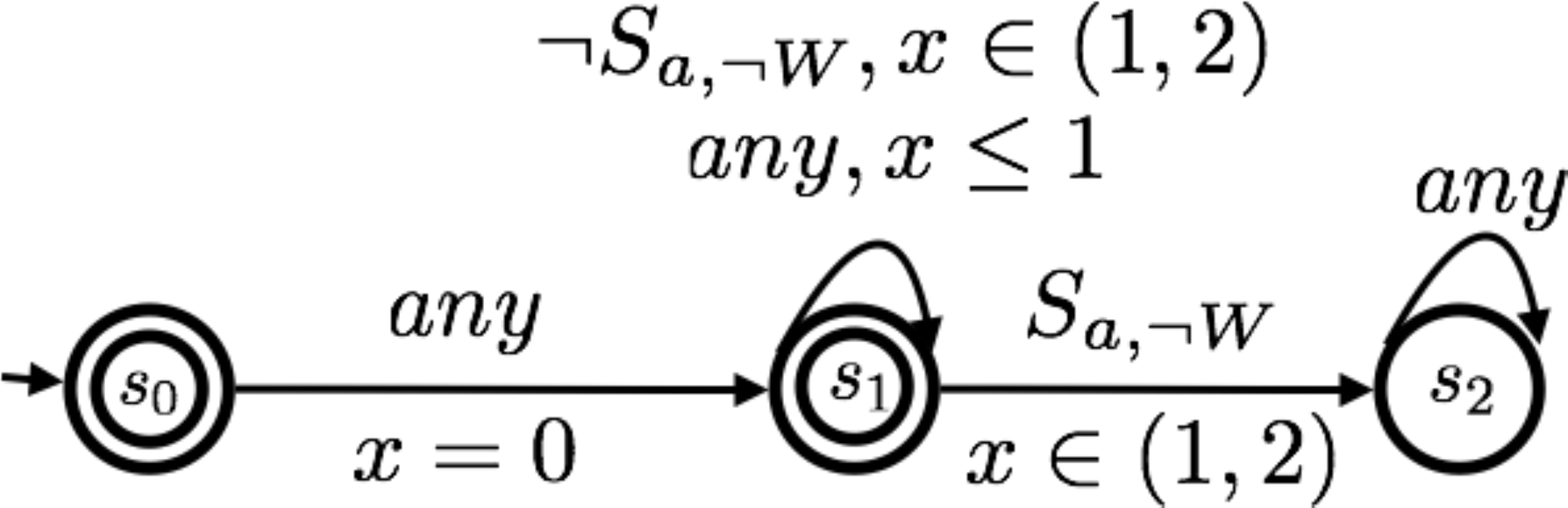} 
	
	\includegraphics[scale=0.25]{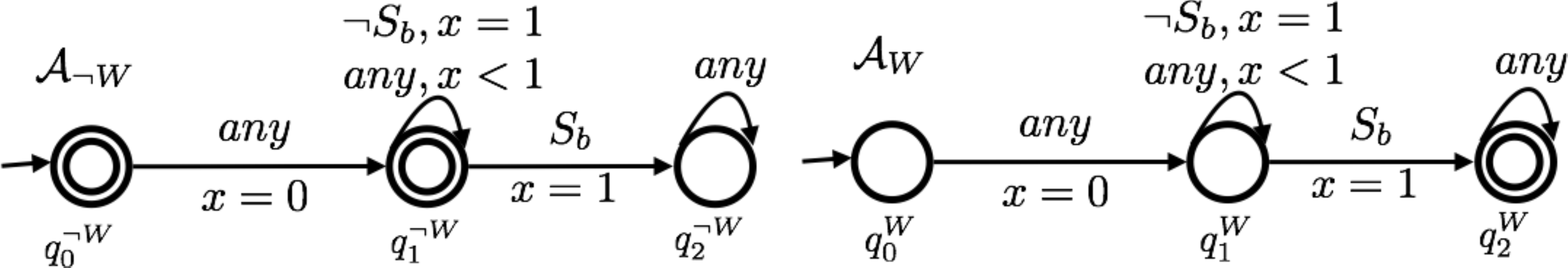}
\caption{$S_b$ stands for any subset of $\Sigma$ containing $b$. $S_{a, \neg W}$ represents any set containing $a$ and $\neg W$. 
On the left is the automaton for $\chi$, and on the right is the one for $\neg W$.}	
\label{eg-class1}
	\end{figure}
	
	One way to obtain the formula now is to convert this automaton to logic, as done in Theorem \ref{aut-tl-1}.
	 To get the formula directly, first we notice that $\neg S_{a, \neg W}$ is a short hand 
	 for $\neg a \vee (a \wedge W)$. The $\regmtl$ formula corresponding 
	 to the DFA over $2^{\Sigma \cup W} \times reg$ is $\reg_{(1,2)} [\neg S_{a, \neg W}]^*$ 
	 which is same as  $\reg_{(1,2)} [\neg a + (a \wedge W)]^*$. Now 
	 the $\regmtl$ corresponding to the DFA for $W$ is $\reg_{[1,1]}(\neg b)^* b$. 
	 Plugging in the formula for $W$, we obtain the formula 
	 $\reg_{(1,2)} [\neg a + (a \wedge \reg_{[1,1]}(\neg b)^* b)]^*$, which can be rewritten as 
	 $\reg_{(1,2)} [a \rightarrow (a \wedge \reg_{[1,1]}(\neg b)^* b)]^*$.

\end{example}

\section{Example Illustrating Section \ref{lem:temp-class-1}}
  \label{app-eg-lem:temp-class-1}
Consider the $\regmtl$ formula 
$\varphi=\regm_{(1,2)}[a \rightarrow \regm_{(0,1)}b^+]$. 
\begin{enumerate}
\item As a first step, we write the $\qkmso$ formula for  $\regm_{(0,1)}b^+$. This is given  by 
$$\zeta_1(t)=\exists t_{first}[t < t_{first}< t+1].\exists t_{last}[t < t_{last}< t+1].
		\forall t'(t < t'< t+1) \varphi(\downarrow t, t_{first}, t_{last}, t')$$
		 where  
		$\varphi(\downarrow t, t_{first}, t_{last}, t')$ is given by \\
			$	((t'=t_{first}) \vee (t'=t_{last}) 
		\vee (t_{first}< t' < t_{last}) \wedge
		\{Q_{b}(t_{first})\wedge Q_{b}(t_{last}) \wedge \\
		\forall t''[t_{first} < t'' < t_{last} \rightarrow Q_b(t'')\})$
\item Next, we rewrite $\varphi$ as  $\psi{=}\regm_{(1,2)}[a {\rightarrow} Z_1]$, where 
$Z_1$ is the witness for $\regm_{(0,1)}b^+$. 
		 
\item 	Since $\psi$ is now a $\regmtl$ formula of modal depth one, we 
have the  $\qkmso$ formula 
$\zeta_0(u)=\exists u_1[u+1 < u_1 < u+2]\varphi(\downarrow u, u_1)$ equivalent to it. 
$$\varphi(\downarrow u, u_1){=}[\neg Q_{a}(u_1) \vee (Q_a(u_1) \wedge \textcolor{red}{Q_{Z_1}(u_1)})] \wedge \forall u''[u+1 < u'' < u+2 \rightarrow (u''{=}u_1)]$$

\item It now remains to plug-in $\zeta_1(u)$ in place of $Q_{Z_1}(u_1)$ in 
$\zeta_0(u)$. Doing this gives 
$\exists u_1[u+1 < u_1 < u+2] \{[\neg Q_{a}(u_1) \vee 
(Q_a(u_1) \wedge \zeta_1(u_1))] \wedge 
\forall u''[u+1 < u'' < u+2 \rightarrow (u''=u_1)]\}$. 

Note that on plugging-in $\zeta_1(u_1)$, the formula obtained is 
in $\qkmso$; further all the bound first order variables 
$v$ are respectively ahead of anchors $u, u_1$. 
\end{enumerate}

\section{Proof of Lemma \ref{lem:frat}}
\label{app:lem:frat}
		 Given a $\wf$  1-clock ATA $\Aa$ over $\Sigma$ with no resets,  
	 we can construct a $\F\regm$ formula, $\varphi$ 
 such that for any timed word $\rho$, $\rho,i \models \varphi$ iff $\rho$, starting from position $i$ has an accepting run in $\Aa$.
Let $q_0$ be the initial location of $\Aa$. Let us consider the case when $q_0 \in Q_{\vee}$. 
		Since $\Aa$  has no  reset transitions, $Q=Q_{\vee}$, and a transition looks like 
		$\delta(q,a)=C_1 \vee \dots \vee C_m$ where each $C_i$ is either a location $q' \in Q$ 
		or a clock constraint $x \in I$. 
A word $w$ is accepted in $\Aa$  when one of the following is true.

\begin{enumerate}
\item Starting from $q_0$, there is a run which reaches some final location $q_f$.  
Let $\re_{q_f,a}$ denote the regular expression (untimed since no clock constraint has been checked this far)
that leads us from $q_0$ to $q_f$, and assume that we enter $q_f$ on reading $a$. 
Then the resultant set of words $w'a$ is accepted where $w' \in L(\re_{q_f,a})$. 
    Disjuncting over all possibilities of final locations $q_f$ and symbols $a \in \Sigma$ we obtain the formula
    $\psi_1=\bigvee_{q_f \in F, a \in \Sigma} \freg_{(0,\infty)}, \re_{q_f,a}(a \wedge \Box \bot)$, where each
    disjunct captures all regular expressions $\re_{q_f,a}$ that guarantee acceptance through $q_f$ when reached on $a$.  
    The $\Box \bot$  ensures that no further symbols are read, and can be written as $\neg \freg_{[0, \infty),\Sigma^*} \top$. 
   
\item The second case is when there is a run which reaches  some location $q$ 
from where, on reading $a$, we choose the disjunct $x \in I_a$ 
in $\delta(q,a)$, and enter an empty configuration. 
Let $\re_{I_a}$ signify the regular expression that collects all words 
$w$ that reach some location $q$ from $q_0$, such that on reading $a$ from $q$, the clock constraint $x \in I_a$ is satisfied. 
Disjuncting over all combinations of intervals and symbols, we obtain the formula 
$\psi_1=\bigvee_{a, I_a} \freg_{I_a, \re_{I_a}}a$. Each disjunct  
 in the formula says that we see an $a$ in the interval $I_a$, 
 and the regular expression $\re_{I_a}$ holds good till that point. 
 Since $\re_{I_a}$ exhaustively collects all words that can reach some location $q$ 
 from where $a$ is read when $x \in I_a$, the formula $\psi_1$ captures all possible 
 ways to accept on enabling a clock constraint. Notice that 
 any suffix can be appended to these set of words, since from the empty configuration, 
 there is no restriction on what can be read. 
\end{enumerate}

The remaining case is when $q_0 \in Q_{\wedge}$. In this case $Q=Q_{\wedge}$. 
If we negate $\Aa$, then we obtain $q_0 \in Q_{\vee}$, and we can apply the case 
discussed above, obtaining a formula in $\fregmtl$ equivalent to the negation of $\Aa$. If we negate this formula, 
we obtain the formula $\psi_2$ equivalent to $\Aa$.  
 		 Any transition to an accepting configuration has to pass 
		 through one of the two cases above. 
		 Thus the formula that we are interested is one of $\psi_1$ or $\psi_2$ depending on whether $q_0 \notin Q_{\vee}$ 
		 or $q_0 \in Q_{\wedge}$.

 \section{Proof of Section \ref{thm:logic-aut2}}
 \label{app:frat-aut}
We prove   for formulae of modal depth 1 first. 
To give an idea, consider $\psi=\freg_{I, \re_0} \psi_0$, a formula of modal depth 1, and having only 
one modality (the $\freg$ modality). We have a DFA $D$ that accepts the regular expression $\re_0$ over some alphabet 
$\Gamma=2^{\Sigma} \backslash \emptyset$. The one-clock ATA $\Aa$ we construct is such that, on reading the first symbol 
of a timed word, we reset $x$ and go to the initial location of the DFA $D$. $D$ continues to run 
until we reach a final location of $D$. While in a final location of $D$ (hence we have witnessed $\re_0$), we check 
if $x \in I$, and if the symbol read currently satisfies $\psi_0$. If so, we accept. Otherwise, we continue running $D$, looking for this combination. 
It is easy to see that the $\Aa$ described here indeed captures $\psi$. 

\begin{lemma}[$\fregmtl$ modal depth 1 to $\wf$-1-ATA-$\lfr$ ]
	Given a $\fregmtl$ formula $\psi$ of modal depth 1,  one can construct a 
	$\wf$-1-ATA-$\lfr$  $\Aa$ such that 
	for any timed word $\rho=$ $(a_1,\tau_1)$ $(a_2,\tau_2)$ $ \dots (a_n,\tau_n)$, 
	$\rho,i \models \psi$ iff $\Aa$ 
	accepts $(a_i, \tau_i) \dots (a_m, \tau_n)$. 
	\label{frat-basic}
\end{lemma}
\begin{proof}
Consider a formula $\psi_1 = \F\regm_{I,\re_0}(\psi_0)$  of modal depth 1, and having only one modality. 
 Clearly, $\re_0$ is an atomic regular expression 
over some alphabet $\Gamma$ and $\psi_0$ is a propositional logic formula over $\Gamma$. Let 
$D = (\Sigma, Q, q_{0}, Q_f,\delta')$ be a DFA such that $L(D)=L(\re_0)$, and $\Sigma=2^{\Gamma} \backslash \emptyset$. 
Given $D$, 
 we now construct the 1-clock ATA  $\Aa=(\Sigma, Q \cup \{q_{init}, q_f\},q_{init},\{q_f\},\delta)$ 
 where $q_{init}, q_f$ are respectively the initial and final locations of $\Aa$, and are disjoint from $Q$. 
 The transitions are as follows. 
\begin{itemize}
	\item $\delta (q_{init},\alpha) = x.q_0$, for all $\alpha \in \Sigma$,	
	\item $\delta (q,\alpha) = \delta'(q,\alpha)$, for all $q \in Q \setminus Q_f$,
		\item $\delta (q,\alpha) = \delta'(q,\alpha) 
		\vee (x.q_f \wedge x\in I)$, 
	for $q \in Q_f$ and  $\alpha$ such that $\alpha \models \psi_0$. For example if $\psi_0=c \wedge d$, then 
	$\alpha \models \psi_0$ iff $\{c,d\} \subseteq \alpha$.  
		\item $\delta (q,\alpha) = \delta'(q,\alpha)$, if 
		$q \in Q_f$ and $\alpha \models \neg \psi_0$, $\delta(q_f,\alpha) = q_f.$
\end{itemize}
On reaching an accepting location of $D$,  if the 
next symbol read (say $\alpha$) satisfies $\psi_0$, we check if the time stamp of $\alpha$ is in the interval $I$. If so, 
we reset $x$ and enter the accepting location $q_f$ of $\Aa$. Once in $q_f$, we always stay in $q_f$. 
If the time stamp of $\alpha$ 
is not in $I$, then we continue running $D$, until we reach again an accepting location of $D$. 
The transitions of $D$ are used until we reach the combination of (i) reaching an accepting location of $D$ along with (ii) the time stamp 
of  the next symbol read is in interval $I$. At this point, we let go of $D$ and accept, by entering $q_f$. 
It is easy to see that $\Aa$ is a $\wf$-1-ATA-$\lfr$.  
Now, we explain how to handle a  boolean combination of $\F \regm$ formulae of modal depth 1 (here, the number 
of modalities are $>1$ though the depth is 1). 
\begin{enumerate}
\item 
It is easy to see that $\wf$-1-ATA-$\lfr$  are closed under complement. On complementation, 
the locations $Q_{\vee}$ and $Q_{\wedge}$ are interchanged, and so are final and non-final locations. 
The argument for correctness of complementation follows as in the general case of 1-clock ATA. 
This takes care of formulae 
of the form $\neg \psi$ when $\psi \in \F \regm$.
\item Consider the case when we start with a formula $\psi_1 \wedge \psi_2$ 
with $\psi_1, \psi_2$ being formulae of modal depth 1 and having only one modality. The above construction 
gives us a $\wf$-1-ATA-$\lfr$ $\Aa_1=(\Gamma, Q_1, q_{01}, q_{f1}, \delta_1)$, 
$\Aa_2=(\Gamma, Q_1, q_{01}, q_{f1}, \delta_1)$,  that are equivalent to $\psi_1, \psi_2$ respectively. 
 We construct the 
1-clock ATA $\Aa$ with locations $Q_1 \cup Q_2 \cup \{q_{init}, \bot\}$, having as 
initial location $q_{init}$ disjoint from $Q_1 \cup Q_2$,  
and  having transitions 
$\delta(q_{init},a)= (x.\delta_1(q_{01},a) \wedge x.\delta_2(q_{02},a)) 
 \vee \bot$,  
$\delta(\bot,a)=\bot$ where $\bot$ is a rejecting location.   Clearly, 
when we start in $\Aa$, we move on to the locations as dictated by $\Aa_1, \Aa_2$
 respectively.  The remaining transitions of $\Aa$ are obtained from $\delta_1, \delta_2$. 
 Since the initial transition respects the $\wf$ condition, and since $\Aa_1, \Aa_2$ are 
 $\wf$-1-ATA-$\lfr$,  we see that $\Aa$ is also a $\wf$-1-ATA-$\lfr$. Acceptance is possible 
 in $\Aa$ only when $\Aa_1, \Aa_2$ simultaneously accept. Thus $L(\Aa)=L(\Aa_1) \cap L(\Aa_2)$. 
\item The case of $\psi_1 \vee \psi_2$ 
with $\psi_1, \psi_2$ being formulae of modal depth 1 and having only one modality follows from 
the fact that we handle complementation and conjunction. 
 	\end{enumerate}
 Thus, we have proved the claim for $\fregmtl$ formula $\psi$ of modal depth 1.
\end{proof}

\subsection*{Lifting to formulae of higher modal depth}

We will induct on the modal depth of the formulae. For the base case, we have the result 
thanks to Lemma \ref{frat-basic}.

Let us assume the result for formulae of modal depth $\leq k$. Consider a formula of modal depth 
$k+1$ of the form $\psi_{k+1} = \F\regm_{I,\re_k} (\psi_k)$ 
 where $\re_k$ is a regular expression over formulae of modal depth $\leq k$
 and $\psi_k$ is a formulae of modal depth $\leq k$.  
For each such occurrence of a smaller depth formula $\psi_i$, let us allocate a witness variable $Z_i$. 
Let $\Zz=\{Z_1, \dots, Z_k\}$ be the set of all witness variables. 
 Let $\Gamma=2^{\Sigma}\backslash \emptyset$. 
Given a subset $S \subseteq \Sigma$ let $\Gamma_S \in S \times 2^{\Zz}$. 
Any occurrence of an element $S \in 2^{\Sigma}$ in $\re_k, \psi_k$ is replaced with 
$\Gamma_S$.  
 At the end of this replacement, $\re_k$ is a regular expression 
over $\Gamma \times 2^{\Zz}$ and $\psi_k$ is a propositional logic formula over $\Gamma \times 2^{\Zz}$. 

Since each $Z_i$ is a witness for a smaller depth formula $\psi_i$, by inductive hypothesis, there is 
a $\wf$-1-ATA-$\lfr$ $\Aa_{Z_i}$  that is equivalent to $\psi_i$. 
Let $\delta_{Z_i}$ be the  transition function of $\Aa_{Z_i}$ and  let $init_{Z_i}$ 
be the initial location  of $\Aa_{Z_i}$.  
We also construct the complement of each such automata $\Aa_{\neg Z_i}$, which has as its transition function 
$\delta_{\neg Z_i}$ and $init_{\neg Z_i}$ as its initial location. 
 
 Lemma  \ref{frat-basic}  gives us 
$\wf$-1-ATA-$\lfr$  (call it $\Cc$)  
over the alphabet  $\Gamma  \times 2^{\Zz}$.  
 Let $\delta_{\Cc}$ denote the transition function of $\Cc$ and let $S_{\Cc}$ be the set of locations 
of $\Cc$. Consider a transition $\delta_{\Cc}(s,\alpha)$ in $\Cc$. 
If $\alpha \in S \times 2^{\Zz} \subseteq \Gamma_S$, for some $S \in 2^{\Sigma}$, then 
the transition $\delta_{\Cc}(s, \alpha)$ is replaced with 
$\delta'(s, S)= \bigvee_{ T \subseteq \Zz}\delta_{\Cc}(s, \alpha) \wedge 
\bigwedge_{\{k \mid Z_k \in T\}} [x.init_{Z_k}] \wedge  \bigwedge_{\{k \mid Z_k \notin T\}} [x.init_{\neg Z_k}]$.
Note that $\Cc$ as well as $\Aa_{Z_i}$ and $\Aa_{\neg Z_i}$ are 
$\wf$-1-ATA-$\lfr$.    
To see why the $\wf$ condition is respected, let $\Cc$ be $\wf$ with $S_{\Cc}$ partitioned into $S_{\vee}$   
and $S_{\wedge}$. If $s \in S_{\vee}$, then 
$\delta_{\Cc}(s,\alpha)$ has the form $C_1 \vee \dots \vee C_m$, and 
conjuncting the reset locations still preserves the form, since these reset locations can be pulled   
into each $C_i$. In case of $s \in S_{\wedge}$ the above procedure does not seem to preserve the conjunctive property of the island. 
 Note that the 1-clock ATA $\Cc$ is a conjunctive island. 
 In this case,  take the negation of $\Cc$, call it $\Cc'$ resulting in a 1-clock ATA which is disjunctive. 
   We then eliminate witnesses using reset transitions as shown above on $\Cc'$, obtaining an  automaton over $\Gamma$. 
   This automaton is then again complemented to get an automaton equivalent to $\Cc$. 

In both cases, let us call the resultant  
1-clock ATA $\Bb'$ over $\Gamma$.  
Clearly, if $\alpha \in S \times T$ is read in $\Cc$, such that $T=\{Z_i, Z_{i_1}, \dots, Z_{i_h}\}$, 
then acceptance in $\Bb'$ is possible iff $\Cc, \Aa_{Z_{i}}, \Aa_{Z_{i_1}}, \dots, \Aa_{Z_{i_h}}$ 
and $\Aa_{\neg Z_j}$ for $j \neq i, i_1, \dots, i_h$ all reach accepting locations on reading the remaining suffix.

\section{The case of higher depth formulae in Section \ref{q2mso-freg}}
\label{app:q2mso-freg} 
Consider a formula $\psi(t_0)$  of metric depth $k+1$. 
 $\psi(t_0)=\Qq_1 t_1 \varphi(\downarrow t_0, t_1)$, such that 
  the metric depth of $\varphi(\downarrow t_0, t_1)$ is at most $k$.  
 We can replace every time constraint sub-formula $\psi_i(t_k)$ occurring in it
by a witness monadic predicate $w_i(t_k)$. This gives a metric depth 1 formula  and we can obtain a $\fregmtl$ formula, say $\zeta$, over 
variables $\Sigma \cup \{w_i\}$ exactly as in the base step. Notice that each $\psi_i(t_k)$ was a formula of modal depth $k$ or less. Hence by induction hypothesis we have an equivalent $\fregmtl$ formula $\zeta_i$. Substituting $\zeta_i$ for $w_i$ in $\zeta$ gives us a formula language equivalent
to $\psi(t_0)$. Since plugging in an $\fregmtl$ formula inside another $\fregmtl$ formula 
results in $\fregmtl$, we obtain the result.
  
    For the case when we start with a $\qtwomlo$ formula of higher depth $k$, we obtain by inductive hypothesis, 
  $\F \sfmtl$ formulae corresponding to each $\psi_i$; secondly, corresponding to the 
  FO formula obtained over the extended alphabet containing $w_i$, we obtain a 
  $\F \sfmtl$ formula using the base case. Plugging in a $\F \sfmtl$ formula in 
  place of the $w_i$ in an $\F \sfmtl$ formula will continue to give a $\F \sfmtl$ formula. 
  
  \section{$\fregmtl$ to \emph{forward} $\qtwomso$ }
 \label{app:lem:freg-q2mso}
Consider a formula $\varphi=\F\regm_{I,\re}\psi$ of modal depth 1. 
 Hence, $\psi$ as well as $\re$ are atomic. 
 Let $\phi_1$ be an MSO sentence equivalent to $\re$.  The formula 
 $\zeta(t)=\texists t' \in t+I(\phi'_1 \wedge Q_{\psi}(t'))$ where $\phi_1'$ is same as $\phi_1$ except that all quantified first order variables 
  $t''$ in $\phi_1'$  lie strictly between $t, t'$ (by semantics of $\freg$, the regular expression is asserted strictly 
  in between), and $Q_{\psi}$ is obtained by replacing all occurrences of $a \in \Sigma$ in $\psi$ with $Q_a$ (if $\psi=a \wedge \neg b$, then 
  $Q_{\psi}(t')=Q_a(t') \wedge \neg Q_b(t')$). It can be seen that $\zeta(t)$ is forward, $\qtwomso$. If we induct on the modal depth, and 
  proceed exactly as in section \ref{lem:temp-class-1}, 
   we obtain a forward, $\qtwomso$ formula equivalent to 
  $\varphi$. Note that the formulae we obtain at each level of $\freg$ will be forward, $\qtwomso$, and 
  plugging in retains this structure.

If we start with a $\F \sfmtl$ formula, then $\phi_1$ will be a FO sentence, and hence $\zeta(t)$ will be forward
$\qtwomlo$. The inductive hypothesis will continue to yield $\qtwomlo$ formulae 
$\zeta(u_1)$ equivalent to $Q_{Z_i}(u_1)$; plugging $\zeta(u_1)$ in place of 
$Q_{Z_i}(u_1)$ in the bigger $\qtwomlo$ formula will hence give rise to a $\qtwomlo$ formula.

\section{Proofs from Section \ref{fixp}}
\label{app:fix}
\subsection{Proof of Lemma \ref{unique_fix_point}}
\begin{proof}
		We prove this using contradiction. Assume that there are two distinct solutions $\alpha$ and $\beta$ for the equation $Z \equiv \varphi (Z)$ with respect to some $\rho$. Thus $\alpha$ and $\beta$ will agree on the truth value of all other propositions except $Z$. Without loss of generality, let $i$ be the last point in the domain of $\alpha$ where $\alpha$ and $\beta$ disagree on the truth value of $Z$. Without loss of generality, we assume $\alpha,i \models Z$ while $\beta, i \models \neg Z$. As both $\alpha$ and $\beta$ are fix point solutions of $Z = \varphi(Z)$, $\alpha,i \models \varphi(Z)$ and $\beta,i \models \neg \varphi(Z)$. 
		Note that as the formulae are guarded, the $Z$ in $\varphi$ will only occur within the scope of a $\reg$ or $\freg$ modality. Both the modalities reason about strict future. That is, the truth value of these modalities depend only on the truth values of propositions at points which are in strict future. Thus the disagreement of $\varphi(Z)$ at point $i$ should imply that the future of $\alpha$ and $\beta$ from the point $i$ is not the same. That is $\alpha[i+1....] \ne \beta [i+1.......]$. This is a contradiction as we assumed that the last point where $\alpha$ and $\beta$ disagree is $i$.
\end{proof}

\subsection{Equivalence of $\mu. \regmtl$ and System of $\regmtl$ Equations \cite{brad}, \cite{igorbrad}}
\label{app:fix2}
We start with an example. 
Consider the formula $\mu Z_1 \circ (\regm_{(1,2)}(a.Z_1.(\nu Z_2\circ (\regm_{(2,3)}(a.Z_2.Z_1.b))).b)$. This could be written as $Z_1 \equiv^\mu \regm_{(1,2)}(a.Z_1.Z_2.b); Z_2 \equiv^\nu \regm_{(2,3)}(a.Z_2.Z_1.b)$.
Thus for any $\mu$ temporal logic formulae $\varphi$, the equivalent system of equations contains as many equations as there are fix point operators in $\varphi$. The simple algorithm of conversion to a system of equations for a given formula $\psi_i$ of the form $\sigma Z_i \circ (\psi(Z_1,\ldots,Z_i))$ will be to reduce it to the equation $Z_i \equiv ^\sigma \psi'(Z_1,\ldots,Z_m)$ where $\psi'$ is obtained from $\psi$ by replacing all its subformulae of the form $\sigma Z_j \circ (\psi(Z_1,\ldots,Z_j))$ with $Z_j$. The set of models accepted by the starting formulae reduced in this way is therefore the solution of $Z_1 \equiv \psi_1$ (that is the solution to the outer most fix point operator).

Similarly, one can also show that any system of such equations can be reduced to the $\mu$ temporal logic formulae. For example, consider the system 
of equations $Z_1 \equiv ^\mu \regm_{(1,2)}(Z_1.Z_2.b^*); Z_2 \equiv^\nu \regm_{(2,3)} (Z_2.Z_1.a^*)$. The solution of the first equation is thus equivalent to,
$\mu Z_1\circ [\regm_{(1,2)}(Z_1.\{\nu Z_2 \circ (\regm_{(2,3)} (Z_2.Z_1.a^*))\} .b^*)]$. 
There is an equivalent formalism with a slightly different syntax in the literature for fixpoints called vectorial fixpoints. The system of equations can also be thought of as a vector of fix point variables simultaneously recursing over the models. 
The reduction from the system of equations (or vectorial fixpoints) to formulae is possible due to Beki\'c Identity \cite{igorbrad}. The blow up is at most exponential.

\end{document}